\documentclass[english]{article}
\usepackage[T1]{fontenc}
\usepackage{babel}
\usepackage{array}
\usepackage{varioref}
\usepackage{fancybox}
\usepackage{calc}
\usepackage{textcomp}
\usepackage{tipa}
\usepackage{amsmath}
\usepackage{amsthm}
\usepackage{amssymb}
\usepackage{graphicx}
\usepackage[authoryear]{natbib}

\makeatletter

\providecommand{\tabularnewline}{\\}

\theoremstyle{plain}
\newtheorem{thm}{\protect\theoremname}
\theoremstyle{plain}
\newtheorem{lem}[thm]{\protect\lemmaname}
\theoremstyle{plain}
\newtheorem{prop}[thm]{\protect\propositionname}

\usepackage{pgfplots}
\pgfplotsset{compat=1.16}

\@ifundefined{showcaptionsetup}{}{%
 \PassOptionsToPackage{caption=false}{subfig}}
\usepackage{subfig}
\makeatother

\providecommand{\lemmaname}{Lemma}
\providecommand{\propositionname}{Proposition}
\providecommand{\theoremname}{Theorem}

\begin{document}
\title{Scenarios for the Transition to AGI\thanks{We would like to thank Ajay Agrawal, Tamay Besiroglu, Erik Brynjolfsson,
Tom Davidson, Joe Stiglitz and seminar participants at Anthropic,
Google, GovAI, and OpenAI for helpful conversations and comments.}}
\date{March 2024}
\author{Anton Korinek and Donghyun Suh\\
University of Virginia, Brookings, and GovAI\\
$\phantom{x}$\\
}
\maketitle
\begin{abstract}
We analyze how output and wages behave under different scenarios for
technological progress that may culminate in Artificial General Intelligence
(AGI), defined as the ability of AI systems to perform all tasks that
humans can perform. We assume that human work can be decomposed into
atomistic tasks that differ in their complexity. Advances in technology
make ever more complex tasks amenable to automation. The effects on
wages depend on a race between automation and capital accumulation.
If the distribution of task complexity exhibits a sufficiently thick
infinite tail, then there is always enough work for humans, and wages
may rise forever. By contrast, if the complexity of tasks that humans
can perform is bounded and full automation is reached, then wages
collapse. But declines may occur even before if large-scale automation
outpaces capital accumulation and makes labor too abundant. Automating
productivity growth may lead to broad-based gains in the returns to
all factors. By contrast, bottlenecks to growth from irreproducible
scarce factors may exacerbate the decline in wages. 
\end{abstract}
\medskip{}


\newpage{}

\section{Introduction}

Recent advances in AI promise significant productivity gains, but
have also renewed fears about the displacement of labor. A growing
number of both AI researchers and industry leaders suggest that it
is time for humanity to prepare for the possibility that we may soon
reach Artificial General Intelligence (AGI) -- AI that can perform
all cognitive tasks at human levels and thus automate them.\footnote{This includes, for example, the CEOs of the three leading AGI labs,
OpenAI's Sam Altman, Google Deepmind's Demis Hassabis, and Anthropic's
Dario Amodei \citep{hasabis23,time2024AGI}. It also includes the
world's most renowned AI researchers, for example two of the godfathers
of deep learning, Geoffrey Hinton and Yoshua Bengio.} This raises a number of fundamental economic questions. What would
the transition to AGI look like? What would AGI imply for output,
wages, and ultimately human welfare? Would wages rise or collapse?

Our paper introduces an economic framework to think about these questions
and evaluate alternative scenarios of technological progress that
may culminate in AGI. Our starting assumption is that human work can
be decomposed into unchanging atomistic tasks that differ in how complex
they are. Advances in technology make ever more complex tasks amenable
to automation. We capture this by assuming that there is a threshold
of task complexity that can be automated at a given time, captured
by an \emph{automation index}. This index grows exogenously over time,
in line with regularities such as Moore's Law. Although our results
hold more broadly, we suggest that in the Age of AI, a natural measure
of task complexity is the amount of compute (shorthand for computational
resources) required for the execution of a task by machines. Some
tasks, such as adding up numbers in a spreadsheet, can be performed
with minimal computation. In contrast, others require a substantial
amount of computation for machines, despite seeming natural and effortless
for humans, such as navigating a bipedal body over an uneven surface.
We describe how tasks differ in computational complexity using a distribution
function that captures tasks in complexity space or, referring to
our preferred interpretation, \emph{tasks in compute space}.\footnote{Although the computational complexity of tasks is most evident for
cognitive tasks, the automation of physical tasks is also greatly
constrained by the computational complexity involved, as captured,
e.g., by Moravec's paradox \citep{moravec88}. We analyze an extension
that explicitly accounts for cognitive and physical tasks in Section
4.}

\begin{figure}
\centering{}\includegraphics[width=0.75\textwidth]{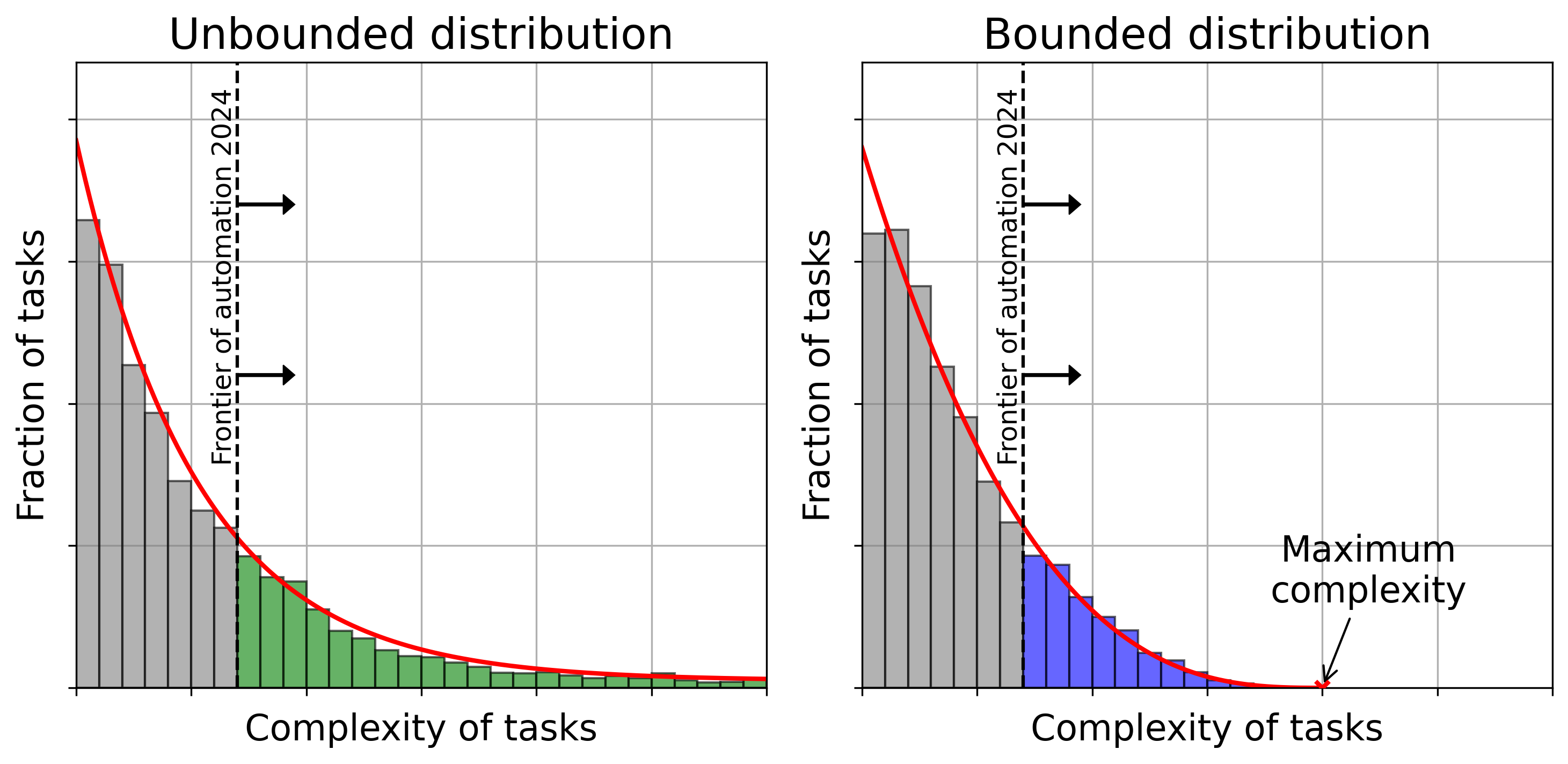}\caption{\label{fig:examples_density}Unbounded and bounded distributions of
tasks in complexity space}
\end{figure}
Throughout the paper, we analyze two opposing cases for the distribution
of tasks in complexity space, which result in sharply different economic
outcomes. First, we consider the possibility that human tasks are
of unbounded complexity, illustrated in the left-hand panel of Figure
\ref{fig:examples_density}. In this case, advances in the automation
index, illustrated by the right-ward movement of the vertical ``frontier
of automation,'' imply that more and more tasks are automated over
time, but that there always remain tasks and by extension jobs that
cannot be automated. Second, we consider a bounded distribution of
task complexity, which reflects that the computational capabilities
of the human brain are finite, as discussed, e.g., in \citet{carlsmith20}.
Bounded distributions result in full automation within finite time
when the frontier of automation crosses the maximum complexity of
tasks performed by humans. An alternative interpretation for tasks
being to complex to automate is that society may choose not to automate
certain tasks even when it is feasible to do so. This may apply, for
example, to some of the tasks performed by priests, judges, or lawmakers.

\medskip{}

We lay out an economic model in which atomistic tasks are gross complements
that are combined to produce final goods. In the spirit of \citet{zeira98}
and \citet{acemoglur18,acemoglur22},\textbf{ }all tasks can be performed
by labor, and automated tasks can be performed by either labor or
capital. However, unlike in the described works, our main focus is
on the edge cases that arise as we come close to full automation. 

Our analysis begins by examining the equilibrium under fixed supplies
of capital and labor. We show that automation can have dramatic impacts
on wages and output even before it reaches all tasks. There exists
a threshold level of the automation index that separates two distinct
regions. As long as the index remains below the threshold, labor remains
scarce relative to capital, and wages remain high. However, once the
automation index surpasses the threshold, the economy enters a second
region, where the scarcity of labor is alleviated, despite the presence
of some tasks that still need to be performed by humans. In this region,
labor and capital become perfect substitutes at the margin so wages
decline starkly to equal the marginal product of capital. The economy
exhibits behavior akin to an $AK$ model.

Next, we characterize the effects of automation on the economy's factor
price frontier (FPF), which reflects all possible combinations of
factor prices that may result from a given level of technology under
different capital/labor ratios. The FPF provides general insights
into the effects of automation that do not depend on specific assumptions
on capital accumulation. We find that for a given level of automation,
wages lie within a bounded interval that expands as the automation
index rises -- but only as long as automation is incomplete. Once
all tasks are automated, the factor price frontier discontinuously
collapses to a single point at which the effective returns to labor
and capital are equalized. For given factor endowments, the effects
of automation on wages are hump-shaped: for low levels of automation,
advances in automation increase wages as the economy becomes more
productive, but for higher levels of automation, wages decline due
to the displacement of labor.

We analyze dynamic settings and show that the effects on wages are
determined by a \emph{race between automation and capital accumulation}.
In addition to the previous two opposing effects on wages from rising
productivity and labor displacement, automation also triggers capital
accumulation that moves the economy up on the factor price frontier,
increasing wages. We characterize an upper bound on output and wages
that is reached in the limit case that the capital stock can instantaneously
adjust to its optimal level whenever automation advances. We show
a powerful analytic result: For any optimizing representative agent
with linearly separable intertemporal preferences, the effects of
automation on output and wages will lie between a lower bound captured
by the constant-capital case and the described upper bound.

When the complexity distribution of tasks is bounded, full automation
is reached in finite time and leads to a collapse in wages, no matter
what savings behavior the representative agent pursues. For unbounded
complexity distributions of tasks, we show that if the tail of remaining
tasks is sufficiently thick, wages will rise forever. By contrast,
if the tail of unautomated tasks is too thin, wages will eventually
collapse.

Next, we simulate a range of scenarios to illustrate our findings
numerically. (Figure \vref{fig:scenarios} shows the main results.)
We start with a ``business-as-usual scenario,'' which captures the
traditional notion that a constant fraction of tasks is automated
each period, similar to \citet{aghional19}. This corresponds to a
Pareto distribution for task complexity together with exponential
growth in the automation index. Since the maximum complexity of tasks
in this scenario is unbounded, true AGI will not be reached in finite
time. In our calibration, both output and wages rise forever in this
scenario, at a pace similar to what advanced countries have experienced
over the past century.

Next we consider two AGI scenarios that span the range of estimates
provided by Geoffrey Hinton, one of the godfathers of deep learning,
who estimated in May 2023 that AGI may be reached within 5 to 20 years---after
declaring that he had ``suddenly switched {[}his{]} views on whether
these things are going to be more intelligent than us.\textquotedblright{}
In our ``baseline AGI scenario'' we assume a bounded task distribution
such that full automation is obtained within 20 year.\footnote{As economists, we hope that the computational complexity of at least
some atomistic tasks that go into writing economics papers is far
into that right tail. Alas, this might be wishful thinking.} In an ``aggressive AGI scenario'' we assume a shorter-tailed distribution
that implies full automation within five years. Our simulation results
imply ten times faster growth than in the business-as-usual scenario,
especially in the aggressive AGI scenario. However, wages collapse
as the economy approaches full automation.

In a fourth scenario, we consider the possibility that there is a
large bout of automation in the near term---for example because AI
rapidly automates cognitive jobs---but that there remains a long
tail of tasks that are harder to automate. As a result of the initial
bout of automation, the economy enters the region in which labor loses
its relative scarcity value, and wages in our simulation collapse.
However, after capital accumulation has caught up sufficiently, labor
becomes sufficiently scarce again so that the economy returns to region
1 and wages rise in line with output growth.

We extend our baseline model to analyze several additional important
considerations. First, we consider the role of fixed factors (such
as minerals or matter) and show that they may pose a bottleneck that
holds back economic growth and worsens the outlook for wages, ultimately
leading to stagnation accompanied by a wage collapse. Next, we add
an innovation sector to analyze the potential for automating technological
progress and show that this lifts the returns of all factors including
wages. We illustrate that sufficient automation may give rise to a
growth singularity whereby output takes off. 

Furthermore, we analyze societal choices to retain certain jobs as
exclusively human even when they can be automated (e.g., priests and
judges), and show that a sufficient volume of such ``nostalgic jobs''
may help to keep labor sufficiently scarce so that wages continue
to grow even when full automation is technically possible. We analyze
the wage-maximizing rate of automation and show that slowing down
automation in an AGI scenario may deliver significant gains to workers
albeit at the cost of forgoing a growing fraction of output.

Next, we evaluate the impact of automation on workers with heterogeneous
skills and susceptibility to being automated. We find that automation
in such a scenario may give rise to an ever-declining fraction of
superstar workers earnings ever-growing wages, whereas the majority
of the labor force is starkly devalued by automation. Finally, we
explore the role of compute as an example of specific capital that
is tailored to automating specific tasks. We observe that in the short
term, compute may earn very high returns, but after an adjustment
period during which sufficient compute has been accumulated (and which
may last long), compute may become just another form of capital that
earns the same return as all other types of capital. 

\paragraph{Related Literature }

The foundational work of \citet{aghional19} explores the impact of
artificial intelligence on economic growth, offering valuable insights
into how technological advances in AI, including AGI, might influence
future economic trajectories. \citet{jones23} underscores the risk
of technological progress, emphasizing existential risk---a concept
crucial in AGI discussions. \citet{davidson23} analyzes a model of
the factors that may lead to a take-off in economic growth if technology
advances near AGI but does not focus on the wage implications. \citet{besiroglual2023}
show how advances in AI may accelerate growth by speeding up R\&D,
and \citet{erdil2023explosive} review the factors by which AGI may
give rise to exponential growth. \citet{trammellk23} provide a useful
survey on the broader implications of advanced artificial intelligence
on economic growth. 

A critical body of literature explores the dynamics between labor
and automation. Seminal works by \citet{acemoglur18,acemoglur22}
and \citet{autor19} provide insights into how automation reshapes
labor markets, focusing on technology as a substitute for individual
worker tasks or how workers and tasks can complement or substitute
for technology. \citet{eloundoual23} and \citet{felten2023languagemodelers}
provide excellent empirical analyses of which tasks are amenable to
automation by the current wave of foundation models. These studies
offer a useful lens for understanding the economic implications of
AI before AGI is reached. Our contribution to these strands of literature
is to look at the limit case of what happens if either all work tasks
are automated or we asymptote towards a world in which all tasks are
automated.

Our paper is also related to a broader literature on AGI and superintelligence
literature. \citet{good1965speculations} was the first to articulate
the potential of an intelligence explosion if AGI is reached. \citet{bostrom14}
provides a comprehensive exploration of superintelligence, highlighting
the potential capabilities of AGI and the profound implications these
might have for society. \citet{yudkowsky2013intelligence} discusses
several of the economic implications of the transition to AGI.

\section{A Compute-Centric Model of Automation\label{sec:A-Compute-Centric-Model}}

\subsection{Tasks in Compute Space}

\shadowbox{\begin{minipage}[t]{0.96\columnwidth}%
\textbf{compute }/k\textschwa m-\textprimstress py\"{u}t/\medskip{}

\textbf{verb:} to determine by calculation

$\phantom{\textbf{verb:}\ }$\emph{The system computed the length
of the shortest path.}

\textbf{noun:} the combined computational resources available for
information pro-

$\phantom{\textbf{noun:}\ }$cessing tasks

$\phantom{\textbf{noun:}\ }$\emph{Modern AI relies on vast amounts
of compute.}\medskip{}

\textbf{Etymology:} Derived from the Latin verb \textquotedbl computare,\textquotedbl{}
meaning \textquotedbl to count, sum up, or reckon together,\textquotedbl{}
the word \textquotedbl compute\textquotedbl{} entered the English
language as a verb in the 16th century. The noun form \textquotedbl compute\textquotedbl{}
gained prominence more recently with the advent of high-performance
digital computers and the increasing need to describe the resources
required for computation.%
\end{minipage}}

\paragraph*{Atomistic Job Tasks}

A central assumption of our analysis is that the work performed by
humans is composed of tasks and sub-tasks -- or unchanging atomistic
task -- that differ in how easily they are automated. In our baseline
model, we focus on cognitive tasks and their potential for automation.
In this setting, an atomistic task is a well-defined computational
assignment that contributes to the accomplishment of a larger job
task. 

These atomistic tasks are fundamental and are significantly smaller
than the tasks that are listed in O{*}Net. Table \ref{tab:top5econ}
lists, for example, the top-5 O{*}Net tasks of economists: to study
data; conduct and disseminate research; compile, analyze and report
data; supervise research; and teach. Each of these O{*}Net ``job
tasks'' involves a wide variety of different atomistic tasks. For
example, the O{*}Net task ``teach theories of economics'' may require
first planning the overall task, recalling different economic theories,
synthesizing a structure, preparing slides, formulating lectures,
synthesizing speech and affect, decoding and responding to student
questions, preparing problem sets, distributing problem sets, grading
problem sets, and so on---all while keeping track of the plan. It
may also require tasks such as recognizing emotional expressions on
students' faces, using theory of mind to evaluate student progress
and dynamically adjust the structure, etc. 

All of these tasks involve a set of basic human brain functions, which
constitute a form of computation. Some of these functions are easily
performed by machines and therefore highly susceptible to cognitive
automation \citep{korinek23llm_wp}, whereas others are more difficult.
What matters for our purposes here is how computation-intensive they
are using machines. 

Recent literature on technology on labor markets observes that innovation
typically gives rise to new job tasks \citep[e.g.,][]{acemoglur18,autor19}.
This holds true when viewed from the perspective of high-level job
tasks such as those captured by O{*}Net. However, when viewed from
an atomistic level that reflects basic brain functions, innovation
merely recombines atomistic tasks in novel ways to produce novel high-level
tasks and jobs. For example, the novel task of ``prompt engineering''
may require atomistic tasks such as defining a desired output, crafting
an initial prompt, entering it, reading the output, evaluating it,
deciding whether to iterate, and finally sharing the output---all
functions that existed long before the invention of generative AI
systems that triggered prompt engineering.

\begin{table}
\begin{itemize}
\item Study economic and statistical data in area of specialization, such
as finance, labor, or agriculture. 
\item Conduct research on economic issues, and disseminate research findings
through technical reports or scientific articles in journals. 
\item Compile, analyze, and report data to explain economic phenomena and
forecast market trends, applying mathematical models and statistical
techniques. 
\item Supervise research projects and students' study projects. 
\item Teach theories, principles, and methods of economics. 
\end{itemize}
\caption{Top-5 Tasks performed by economists (O{*}Net database)\label{tab:top5econ}}
\end{table}

\paragraph*{Task Complexity and Compute Intensity }

Our baseline model emphasizes differences in complexity as a key dimension
when studying the automation potential of tasks. Our preferred interpretation
for what makes tasks difficult to automate is their compute intensity,
which refers to the amount of computational resources required to
perform a specific task. Compute intensity can easily be measured
by the amount of floating point operations (FLOP) that need to be
executed to perform a given task.  The computational complexity for
machines to execute a task often differs starkly from how easy or
difficult it is for humans.\footnote{For example, \citet{moravec88} observed that some tasks that \emph{feel}
easy to execute for humans, such as vision processing, employ a large
amount of dedicated grey matter and are, in fact, computationally
quite intensive. There are also certain tasks that require very little
compute in dedicated machines but that are difficult for the human
brain since it has not evolved for them: for example, arithmetic operations
take just one FLOP on a basic computer but require significant amounts
of grey matter for human brains to perform, likely involving the equivalent
of billions of FLOP.} Still, it is the computational complexity for machines that determines
whether a task can be automated.

\paragraph{Advances in Computing}

One of the main drivers of recent advances in AI has been the increased
availability of computing power. Moore's Law, first described by Gordon
\citet{moore65}, describes that the performance of cutting-edge computer
chips doubles approximately every two years. The regularity has held
for the past sixty years. Additionally, the amount of compute deployed
in cutting-edge AI systems has grown even faster over the past decade,
doubling roughly every six months, as shown in \citet{sevillah22}
and depicted in Figure \ref{fig:computegrowth}. Improvements in algorithms
have further accelerated the growth in capabilities of cutting-edge
AI systems \citep{besiroglual2023}. 

For our analysis below, we assume that there is an automation index
that captures the maximum complexity of tasks that can be automated.
This index grows exponentially at an exogenous rate, mirroring the
type of advances captured by Moore's Law and Figure \ref{fig:computegrowth}.
As the automation index increases, a growing mass of tasks can be
automated.

\begin{figure}
\noindent \begin{centering}
{\small{}\includegraphics[width=0.75\columnwidth]{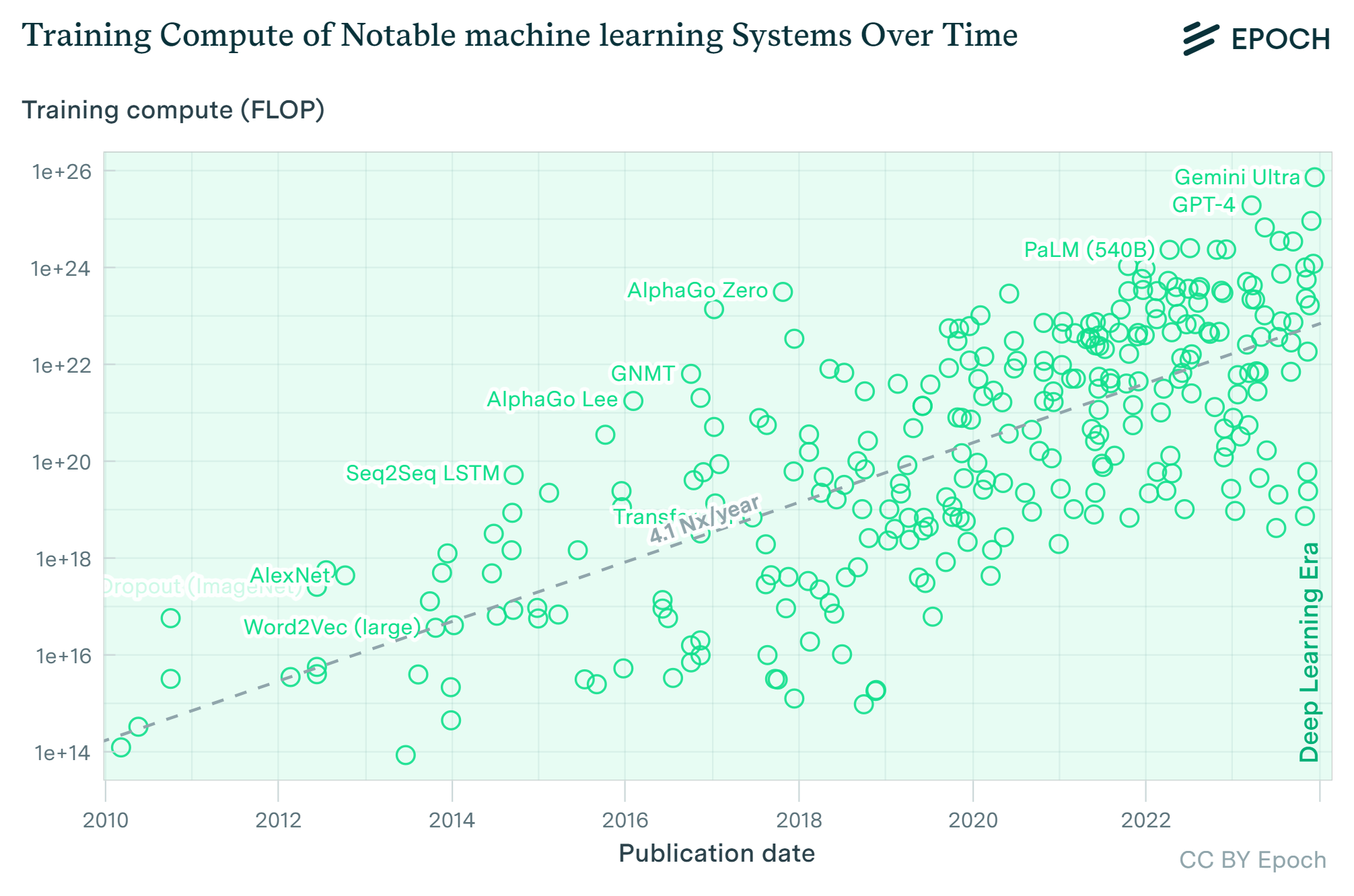}}{\small\par}
\par\end{centering}
{\small{}\caption{\label{fig:computegrowth}Training compute of frontier AI systems
over time\protect \\
(Copyright \textcopyright{} 2024 by Epoch under a CC-BY-4.0 license;
\citealp{sevillah22}.)}
}{\small\par}
\end{figure}

\subsection{Baseline Model }

Consider a representative household in a static economy who is endowed
with $L=1$ units of labor and $K>0$ units of capital. There is a
continuum of tasks that differ in their computational complexity $i$.
The distribution function $\Phi(i)$ reflects the cumulative mass
of tasks with complexity $\leq i$ and satisfies $\Phi(0)=0$ and
$\lim_{i\rightarrow\infty}\Phi(i)=1$. If the distribution function
is differentiable, we call its derivative $\phi(i)$ the density of
tasks of complexity $i$. Examples are shown in Figure \ref{fig:examples_density}.

To produce aggregate output $Y$, we combine all the tasks of different
complexity using a CES aggregator with elasticity of substitution
$\sigma$. 
\begin{equation}
Y=A\left[\int_{i}y(i){}^{\frac{\sigma-1}{\sigma}}d\Phi(i)\right]^{\frac{\sigma}{\sigma-1}}\label{eq:Y}
\end{equation}
where $y(i)$ is the amount of type $i$ tasks employed in the production
of output. We generally assume $\sigma<1$, reflecting that the atomistic
tasks are gross complements.

Each task is performed using capital $k(i)$ and labor $\ell(i)$
according to the production function
\begin{equation}
y(i)=a_{K}(i)k(i)+a_{L}(i)\ell(i)\label{eq:y(i)}
\end{equation}
where the coefficients $a_{K}(i)$ and $a_{L}(i)$ reflect the efficiency
of capital and labor. We assume that the exogenous index $I$ reflects
the state of automation and defines a complexity threshold such that
all tasks below the threshold can be performed with either capital
or labor but all the tasks above the threshold require labor. We normalize
the technological parameters $a_{K}\left(i\right)=a_{L}\left(i\right)=1$
except that $a_{K}\left(i\right)=0$ if $i\geq I$. In other words,
\[
y\left(i\right)=\begin{cases}
k\left(i\right)+\ell\left(i\right) & \text{for }i<I\\
\ell\left(i\right) & \text{for }i\geq I
\end{cases}
\]

\paragraph*{Strategies}

The representative agent supplies her endowments of labor and capital
every period at the prevailing factor prices $w$ and $R$ and makes
no interesting economic decisions. The representative firm in the
economy maximizes profits by hiring capital $k\left(i\right)$ and
labor $\ell\left(i\right)$ for each task at the prevailing factor
prices $w$ and $R$ to produce $y\left(i\right)$, which is then
combined to produce final output. The firm's maximization problem
is
\[
\max_{k(i),\ell(i)}Y-R\int_{i}k(i)d\Phi(i)-w\int_{i}\ell(i)d\Phi(i)\quad\text{s.t.}\quad(1),(2)
\]

\paragraph{Equilibrium}

An equilibrium in the baseline model consists of a set of $\left\{ k\left(i\right),\ell\left(i\right),y\left(i\right)\right\} _{i\geq0}$
and factor prices $w$ and $R$ such that the representative firm
solves its maximization problem and markets for capital and labor
clear, i.e., 
\[
\int_{i}k(i)d\Phi(i)=K\qquad\int_{i}\ell(i)d\Phi(i)=1
\]
Since there are no market imperfections, the described equilibrium
also constitutes the first-best of the economy.

\subsection{Equilibrium: Characterizing Two Regions}

\paragraph*{Scarcity of Labor}

For given factor endowments $\left(K,L\right)$, there are two possible
regimes for the scarcity of labor, depending on the level of the automation
index $I$: If the index is low enough so that labor is relatively
scarce, then the return on labor is greater than the return on capital,
$w>R$, and the scarce labor is employed solely in those tasks that
cannot be automated. 

Conversely, if the state of automation is sufficiently advanced that
only a small fraction of tasks are exclusive to human labor, then
$w=R$ holds, and labor is employed not only in the remaining unautomated
tasks but also in some of the automated tasks. At the margin, capital
and labor are perfect substitutes. 
\begin{lem}[Scarcity of labor]
\label{lem:scarcity-of-labor}For given $(K,L)$, there is a threshold
value for the state of automation $\hat{I}$ that is defined by 
\begin{equation}
\Phi\left(\hat{I}\right)=\frac{K/L}{1+K/L}\label{eq:Ihat}
\end{equation}
and increasing in the $K/L$-ratio such that there are two regions:

\noindent \textbf{Region 1:} If $I<\hat{I}$, then labor is scarce
compared to capital. In this regime, labor is employed only for tasks
with $i>I$. Output is 
\begin{equation}
Y=F\left(K,L;I\right)=A\left[K^{\frac{\sigma-1}{\sigma}}\Phi(I)^{\frac{1}{\sigma}}+L{}^{\frac{\sigma-1}{\sigma}}(1-\Phi(I))^{\frac{1}{\sigma}}\right]^{\frac{\sigma}{\sigma-1}}\label{eq:Yeff-case1}
\end{equation}
and wages satisfy
\[
w=A^{\frac{\sigma-1}{\sigma}}\left(Y/L\right)^{\frac{1}{\sigma}}\cdot(1-\Phi(I))^{\frac{1}{\sigma}}>R
\]

\noindent \textbf{Region 2:} If $I\geq\hat{I}$, then the relative
scarcity of labor is relieved, and labor earns the same return as
capital $w=R=A$; if the inequality is strict, some labor is deployed
alongside capital for tasks with $i<I$, and labor and capital are
perfect substitutes for the marginal task. Output is given by the
linear function
\begin{equation}
Y=F\left(K,L\right)=A\left(K+L\right)\label{eq:Yeff-case2}
\end{equation}

Conversely, for given $I$, there is a threshold $\kappa\left(I\right)=\Phi(I)/\left[1-\Phi\left(I\right)\right]$
such that the economy is in region 1 if $K/L>\kappa\left(I\right)$
and in region 2 if $K/L\leq\kappa(I)$. The threshold $\kappa(I)$
is increasing in $I$, i.e., if $K/L$ is marginally above the threshold,
further automation pushes the economy from region 1 into region 2
where the scarcity of labor is relieved.
\end{lem}

\begin{proof}
Assume first that all labor is employed in tasks with $i\geq I$ and
observe that the symmetry of the production function across all tasks
implies that an identical amount of capital $k=K/\Phi(I)$ will be
employed in each task below the threshold and an identical amount
of labor $\ell=L/(1-\Phi(I))$ for each task above the threshold for
given aggregate $K$ and $L$. The production function can then be
written as
\begin{align*}
Y=F\left(K,L;I\right) & =A\left[k^{\frac{\sigma-1}{\sigma}}\Phi(I)+\ell{}^{\frac{\sigma-1}{\sigma}}(1-\Phi(I))\right]^{\frac{\sigma}{\sigma-1}}\\
 & =A\left[K^{\frac{\sigma-1}{\sigma}}\Phi(I)^{\frac{1}{\sigma}}+L^{\frac{\sigma-1}{\sigma}}(1-\Phi(I))^{\frac{1}{\sigma}}\right]^{\frac{\sigma}{\sigma-1}}
\end{align*}
proving equation (\ref{eq:Yeff-case1}).

The firm's optimization problem implies the first-order conditions
\begin{align}
\ensuremath{F_{K}} & =A^{\frac{\sigma-1}{\sigma}}Y^{\frac{1}{\sigma}}\cdot K^{-\frac{1}{\sigma}}\Phi(I)^{\frac{1}{\sigma}}=R\label{eq:R}\\
F_{L} & =A^{\frac{\sigma-1}{\sigma}}Y^{\frac{1}{\sigma}}\cdot L^{-\frac{1}{\sigma}}(1-\Phi(I))^{\frac{1}{\sigma}}=w\label{eq:w}
\end{align}

By comparing these two expressions, we can see that the return on
capital $R$ is less than the return on labor $w$ as long as $K^{-\frac{1}{\sigma}}\Phi(I)^{\frac{1}{\sigma}}<L^{-\frac{1}{\sigma}}(1-\Phi(I))^{\frac{1}{\sigma}}$
or, equivalently, $k>\ell$, i.e., the capital assigned to each automated
task is greater than the labor assigned to unautomated tasks. Expressing
this in terms of aggregate supplies of factors, the condition is 
\begin{equation}
\frac{K}{L}>\kappa(I):=\frac{\Phi(I)}{1-\Phi(I)}\label{eq:threshold_condition}
\end{equation}
The right-hand side is an increasing function of $I$ that goes from
$0$ to $\infty$. By implication, there is a value of $I$ such that
the inequality is violated for all $I>\hat{I}$. 

When that threshold is crossed, it is more efficient to allocate some
labor to tasks with $i<I$, and the marginal unit of labor is perfectly
substituable with capital. By implication, $k\left(i\right)=\ell\left(i\right)=K+L$,
and the CES aggregator simplifies to equation (\ref{eq:Yeff-case2}).
Alternatively,  the threshold for $\Phi(I)$ can be expressed explicitly
as an increasing function of the $K/L$-ratio by solving for
\[
\Phi\left(\hat{I}\right)=\frac{K/L}{1+K/L}
\]
The remaining results stated in the lemma follow immediately.
\end{proof}
Intuitively, region 1 reflects the world as we have experienced it
over the past 200 years, in which capital and labor are complementary
in production, and labor is comparatively scarce. Figure \ref{fig:water_levels}
shows that for given factor supplies, a higher automation index $I$
increases the mass of tasks that can be accomplished with capital,
implying that the available capital is spread over a greater number
of tasks and becomes scarcer. Conversely, automation reduces the mass
of tasks that is exclusive to labor, implying that the available labor
can be concentrated on fewer tasks and becomes less scarce. As the
automation index reaches the threshold $\hat{I}$, there are so few
tasks left that are exclusive to labor that labor no longer enjoys
a scarcity advantage over capital, and the returns on the two factors
are equated.

Note that the threshold $\hat{I}$ depends solely on relative factor
supplies, not on the elasticity of substitution $\sigma$ between
capital and labor. As soon as labor is no longer scarce, it will be
used interchangably with capital in the marginal task, and this holds
even when individual tasks are highly complementary as reflected by
low values of the elasticity of substitution (as long as $\sigma>0$).

\begin{figure}

\includegraphics[width=1\textwidth]{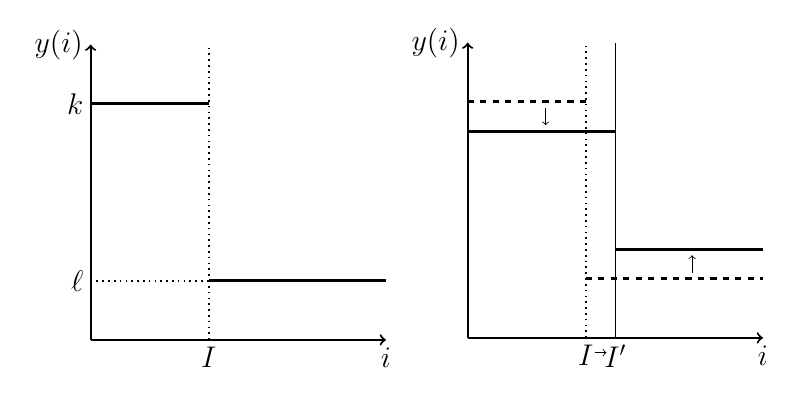}\caption{Automation and the scarcity of labor}
\label{fig:water_levels}

\end{figure}

\subsection{Factor Price Frontier (FPF)}

For an analysis of the effects of advances in automation $I$ on factor
returns, let us characterize the factor price frontier associated
with the firm's technology. The factor price frontier depicts all
possible combinations of factor prices $R$ and $w$ that will result
from different proportions of factor supplies $K$ and $L$ in a competitive
economy with profit-maximizing firms under a given technology. 
\begin{lem}[Factor Price Frontier (FPF)]
\label{lemma1}For a given automation index $I$, the factor price
frontier slopes downwards, starting from a limiting point $w^{\ast}(I)=A(1-\Phi(I))^{\frac{1}{\sigma-1}}$
and $R=0$ as $K/L\rightarrow\infty$ to the point $w=R=A$ when $K/L\leq\kappa(I)$.
Increases in $A$ move the FPF proportionately outwards. Increases
in $I$ raise $w^{\ast}(I)$ and swivel the factor price frontier
clock-wise. 
\end{lem}

\begin{proof}
We obtain the factor price frontier from the aggregate cost function,
which is the dual of the aggregate production function. The associated
unit cost function represents the minimum cost at which a competitive
optimizing firm can produce one unit of final output, given factor
prices $w$ and $R$. In the region of $I<\hat{I}$, the unit cost
function associated with equation (\ref{eq:Yeff-case1}) is
\[
C(w,R;I)=\frac{1}{A}\bigg(R^{1-\sigma}\Phi(I)+w^{1-\sigma}(1-\Phi(I))\bigg)^{\frac{1}{1-\sigma}}
\]
Since we employed the final good as the numeraire good, this cost
function needs to equal 1 in a competitive economy. The factor price
frontier when $I<\hat{I}$ is thus given by all pairs of $(w,R)$
that satisfy the equation $C\left(w,R;I\right)=1$, or equivalently,
\begin{equation}
w=\bigg(\frac{A^{1-\sigma}-R^{1-\sigma}\Phi(I)}{1-\Phi(I)}\bigg)^{\frac{1}{1-\sigma}}\label{eq:wFPF}
\end{equation}
Asymptotically, as $K/L$ goes to infinity, we can see from equations
(\ref{eq:R}) and (\ref{eq:w}) that the return to capital $R$ goes
to zero, whereas the wage converges to
\begin{equation}
w^{\ast}(I)=\lim_{K/L\rightarrow\infty}w=A\left[0\cdot\Phi(I)^{\frac{1}{\sigma}}+(1-\Phi(I))^{\frac{1}{\sigma}}\right]^{\frac{1}{\sigma-1}}(1-\Phi(I))^{\frac{1}{\sigma}}=A(1-\Phi(I))^{\frac{1}{\sigma-1}}\label{eq:limit-wage}
\end{equation}
Conversely, when $I\geq\hat{I}$, the cost function is simply $C\left(w,R;I\right)=\min\left\{ w,R\right\} /A$,
and the factor price frontier is degenerate and consists of a single
point $w=R=A$.
\end{proof}

\begin{figure}
\centering{}\includegraphics[width=0.42\textwidth]{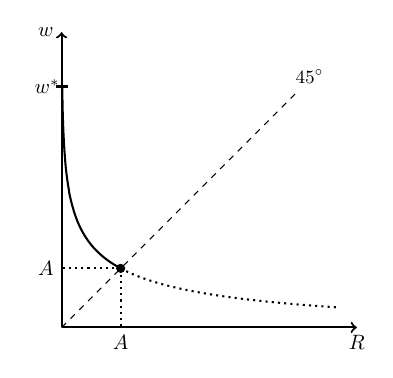}\includegraphics[width=0.42\textwidth]{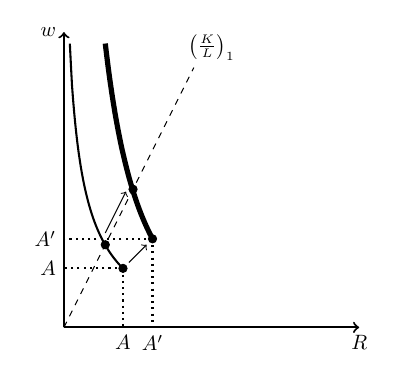}\caption{\label{fig:FPF_risingA}Factor price frontier and its dependence on
$A$}
\end{figure}
The factor price frontier is illustrated in Figure \ref{fig:FPF_risingA}.
The area above the 45 degree line corresponds to Region 1 of Lemma
1, reflecting a high capital-labor ratio $K/L>\kappa(I)$ and $w>R$.
Higher capital intensity $K/L$ moves factor returns up and to the
left along the frontier, i.e., it increases $w$ and reduce $R$.
Conversely, when $K/L\leq\kappa(I)$, we enter Region 2 of the lemma,
and the factor price frontier corresponds to a single dot on the 45
degree line at which $w=R=A$.

The right panel of Figure \ref{fig:FPF_risingA} shows how an increase
in the level of technology $A$ pushes out the factor price frontier
-- for any ratio of $K/L$, it scales the returns of all factors
proportionately. This exemplifies how the factor price frontier serves
as a convenient tool to describe how factor returns are impacted by
technological changes across any levels of factors supplies. 

\subsubsection*{The Automation Path on the Factor Price Frontier}

We next turn to the effects of automation for a given capital stock
$K$ or equivalently, capital intensity $k=K/L$. Then it is easy
to see that:
\begin{lem}[Automation and Output]
An increase in automation $d\Phi(I)$ raises output as long as $I<\hat{I}$,
and leaves output unaffected otherwise.
\end{lem}

\begin{proof}
For $I<\hat{I}$, the result follows by differentiating expression
(\ref{eq:Yeff-case1}), 
\begin{align*}
\frac{dY}{d\Phi(I)}= & \frac{1}{\sigma-1}A^{\frac{\sigma-1}{\sigma}}Y{}^{\frac{1}{\sigma}}\cdot\bigg(k^{\frac{\sigma-1}{\sigma}}-\ell^{\frac{\sigma-1}{\sigma}}\bigg)
\end{align*}
Given $\sigma<1$, the derivative is positive as long as $k>\ell$,
which is the condition for being in region 1 in Lemma \ref{lem:scarcity-of-labor}
in which the production function is relevant. For $I\geq\hat{I}$,
the relevant production function is (\ref{eq:Yeff-case2}), which
is independent of $I$.
\end{proof}
Intuitively, for output to rise, capital must be sufficiently abundant,
delivering a productivity gain from deploying the amply available
capital to a greater number of tasks. This is frequently termed the
\emph{productivity effect} of automation. 

Let us look at factor returns next.
\begin{lem}[Automation and Factor Returns]
\label{lem:automation_wages}\label{lem:automation-and-wages}(i)
An increase in automation $d\Phi(I)$ always raises $R$ as long as
$I<\hat{I}$. The effect on $w$ is hump-shaped: there is a threshold
$I^{\ast}(K/L)$ with $\Phi(I^{\ast}(\cdot))\in(0,1)$ such that wages
$w$ rise in $\Phi(I)$ as long as $I<I^{\ast}(K/L)$ or, equivalently,
as long as $K/L>\kappa^{\ast}(I)$, but decline in $\Phi(I)$ for
$I>I^{\ast}(K/L)$ or, equivalently, $K/L<\kappa^{\ast}(I)$.

(ii) For $\Phi(I)=0$, the return on capital is $R=0$, and wages
equal $w=A$. For $\Phi(I)\geq\kappa/(1+\kappa)$, both equal $R=w=A$.
The latter condition always holds if $\Phi(I)=1$. 
\end{lem}

\begin{proof}
The limit results follow readily from equations (\ref{eq:R}) and
(\ref{eq:w}) and from the second part of Lemma 1. By differentiating
equation (\ref{eq:R}), with respect to $\Phi(I)$, we can see that
automation always raises the return on capital.

To see how automation affects wages, consider the derivative of $\log w$
with respect to $\Phi$ from the firm's optimality condition (\ref{eq:w}):
\begin{align}
\frac{d\log w}{d\Phi(I)}= & \frac{1}{\sigma-1}\frac{1}{\sigma}\bigg(k^{\frac{\sigma-1}{\sigma}}-\ell^{\frac{\sigma-1}{\sigma}}\bigg)(Y/A)^{\frac{1-\sigma}{\sigma}}-\frac{1}{\sigma}\frac{1}{1-\Phi(I)}.\label{eq:automation-wages}
\end{align}

The first term reflects the productivity effect of automation, which
is positive under condition (\ref{eq:threshold_condition}), reflecting
that producing the marginal task using a relatively more abundant
$k$ units of capital rather than a scarce $\ell$ units of labor
increases output. The second term captures the \emph{displacement
effect }of automation and reduces labor income. It reflects that the
labor used in each unautomated task $\ell=L/(1-\Phi)$ increases,
as captured by the term in the denominator, thereby pulling down the
marginal product of labor. 

As $I$ rises, wages rise at first---at $\Phi(I)=0$ we find $\frac{d\log w}{d\Phi(I)}=\frac{1}{1-\sigma}>0$.
As $I$ becomes larger, the first term in (\ref{lem:automation_wages})
declines, reaching zero for $I=\hat{I}$, and the absolute value of
the second term grows and eventually dominates the first term---in
the limit of $\Phi\rightarrow1$, the second term becomes infinitely
large. Thus, there exists an intermediate value $I^{\ast}(K/L)$ after
which further automation reduces wages. Notice that the first term
is increasing in $K/L$ whereas the second term is independent of
$K/L$. The threshold can alternatively be expressed as $K/L<\kappa^{\ast}(I)$.

\begin{figure}
\centering{}\includegraphics[width=0.5\textwidth]{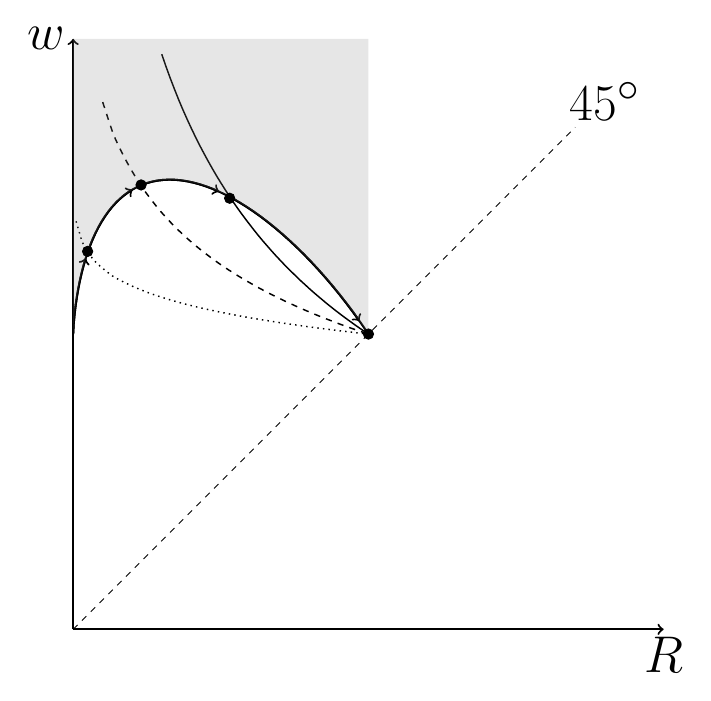}\caption{\label{fig:FPF_risingAut}Factor price frontier and automation}
\end{figure}
\end{proof}
Figure \ref{fig:FPF_risingAut} illustrates that an increase in automation
$I$ ``rotates'' the factor price frontier clockwise, for example,
from the dotted to the dashed and solid lines. If the economy is in
the labor-scarce region 1 (above the 45 degree line), automation raises
wages for a given return of capital and also the maximum wage level
$w^{\ast}(I)$. For given $K/L$, the path of factor prices that
results from rising automation $I$ is illustrated by the hump-shaped
bold line with arrows in the figure. Along the path, $R$ rises continually
whereas $w$ at first rises but eventually falls. When automation
reaches $\hat{I}(K/L)$, the economy ends up in the degenerate equilibrium
with $w=R=A$ on the 45-degree line.

\subsection{Automation and Factor Earnings}

Figure \ref{fig:risingPhi} shows the effects of automation on total
output for given factor supplies as well as its split into the wage
bill and the total returns to capital. The horizontal axis depicts
the fraction $\Phi\left(I\right)$ of automated tasks, which goes
from zero to one. The left-hand panel illustrates the case of equal
capital and labor endowments, $K=L=1$, and modest complementarity
with an elasticity of substitution $\sigma=0.5$ between the two.
As long as the economy is in the scarce-labor region (Region 1), output
is a strictly monotonic function of automation. At first, automation
almost exclusively benefits labor, and the returns to capital are
minuscule. But as automation increases and we come closer to Region
2, the wage bill reaches a ceiling and starts to decline. Further
automation still raises output, but the returns to capital grow faster
than output, at the expense of the wage bill. When Region 2 is reached
at $\Phi=0.5$, both factors earn equal returns. Given equal endowments,
this translates into capital and labor shares of one-half each. 

\begin{figure}
\centering{}\subfloat[Equal factor endowments]{\includegraphics[width=0.48\columnwidth]{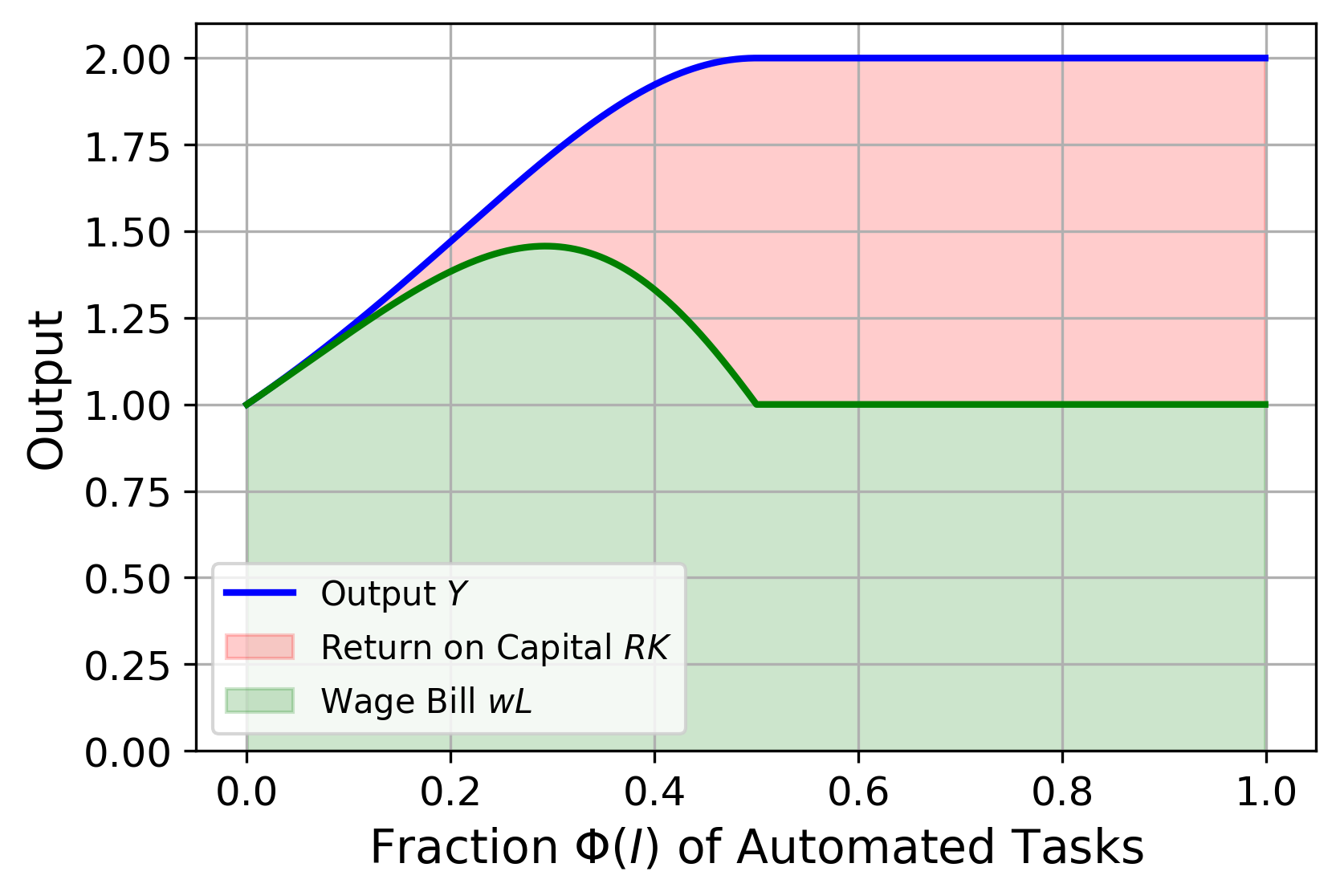}}\subfloat[Abundant effective capital]{\includegraphics[width=0.48\columnwidth]{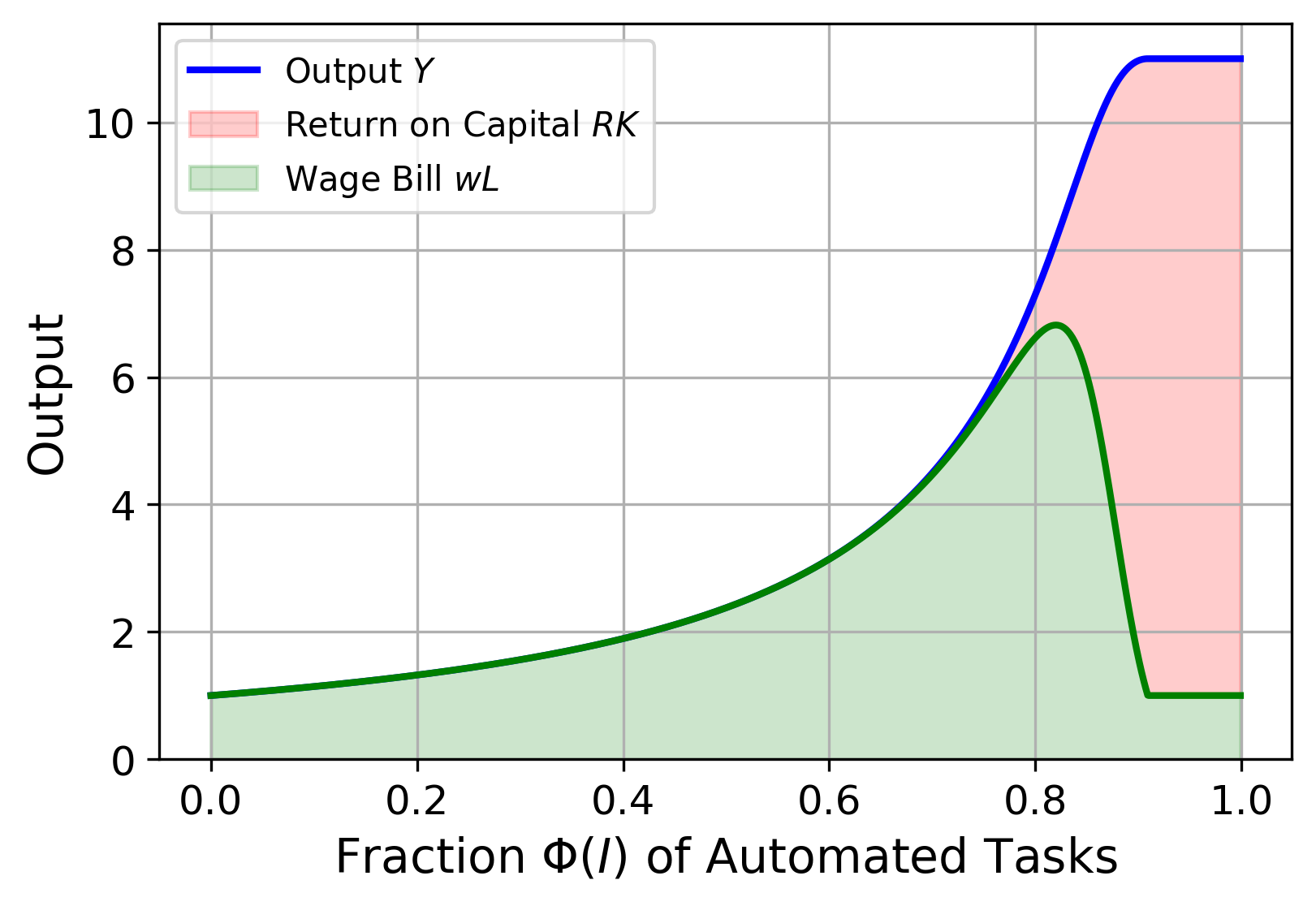}}\caption{\label{fig:risingPhi}Static equilibria under rising automation}
\end{figure}

The right panel of the figure shows an alternative scenario in which
the effective supply of capital is ten times higher than labor, i.e.,
$L=1$ and $K=10$, and in which the two are strong complements with
$\sigma=0.2$. The abundance of capital and the strong complementarity
imply that the region in which most of the benefits go to labor is
much larger, but so is the drop in the wage bill once a critical threshold
is surpassed: whereas wages \emph{seem }to be growing exponentially
in $\Phi\left(I\right)$ up until $\Phi\approx0.80$, they experience
a precipitous decline by about 85\% starting around $\Phi\approx0.83$,
accompanied by a meteoric rise in the returns to capital. When Region
2 is reached at $\Phi=10/11$, factor returns are equalized, and given
the relative factor endowments, the capital share of the economy is
ten times the labor share. This example highlights that the fate of
labor can change rapidly when certain thresholds are crossed.

Crucially, the effect of automation on output---and per-capita income---depends
on the capital available. In the illustration in the left panel, full
automation merely doubles output; in the right panel, output grows
eleven-fold. This observation naturally leads us to the next step
of our analysis---to analyze how automation interacts with capital
accumulation in a dynamic setting.

\section{Dynamics: The Race between Automation and Capital Accumulation}

The dynamics of output and wages depend not only on technological
advances---captured by the automation index $I$---but also on capital
accumulation and by extension on the savings behavior of the agents
in the economy. This section analyzes these forces in a dynamic setting.

\subsection{Automation Scenarios}

\paragraph{Progress in automation}

We assume that the automation index $I$ grows exponentially over
time at an exogenous rate of $g$, reflecting Moore's law and similar
regularities. For an initial $I_{0}$, the time path of $I$ (omitting
the time index $t$ for conciseness) is given by
\[
I=I_{0}e^{gt}
\]
We can equivalently write that $\log I=\log I_{0}+gt$ grows linearly
at the rate $g$. 

We consider different distributions $\Phi\left(i\right)$ of task
complexity or tasks in compute space to capture alternative scenarios
for the advent of AGI:

\paragraph{Business-As-Usual Scenario (Unbounded Distribution)}

We model unbounded complexity distributions of tasks $\Phi\left(i\right)$
as Pareto, implying that $\log i$ is described by an exponential
distribution, $\log i\sim Exp(\lambda)$ with decay parameter $\lambda$.
The resulting cumulative distribution function is $\Phi(i)=1-e^{-\lambda\log i}$.
If the automation index $I$ grows exponentially at rate $g$, the
fraction of non-automated tasks declines at rate $\lambda\cdot g$.
This distribution has an infinite right tail, meaning that there will
always be tasks that cannot be automated.

\paragraph{Baseline and Aggressive AGI Scenarios (Bounded Distributions)}

For our AGI scenarios, we assume a bounded complexity distribution
of tasks to capture the scenario that the tasks that can be performed
by human brains is limited by an upper bound so automation crosses
the threshold $\hat{I}$ within finite time. We assume that $\Phi(i)$
follows a power function $\Phi(i)=1-(1-\log i/\log I^{\max})^{\beta}$
with $\beta=1$ and with normalization $I^{\max}=I_{0}e^{gT}$ such
that all tasks are automated after $T$ years.\footnote{Power function distributions are a special case of beta distributions
with a beta parameter $\alpha=1$.} Following Hinton's predictions, we set $T=20$ in the baseline AGI
scenario and $T=5$ in the aggressive AGI scenario. For $I>I^{\max}$,
we keep $\Phi(i)=1$ capturing full automation.

\paragraph{Bout of Automation (Mixed Distribution)}

We consider a fourth scenario in which rapid advances in AI automate
a large fraction of tasks within a short time span, but in which we
assume that there remains an unbounded tail of tasks that cannot be
automated, for example, because of legal or cultural reasons. Analytically,
we assume a mixture of the two scenarios above. Specifically, $\Phi(i)$
is defined as $\Phi(i)=\omega\bigg[1-\left(1-\log i/\log I^{\max}\right)^{\beta}\bigg]+(1-\omega)\bigg[1-e^{-\lambda\log i}\bigg]$
where $\omega\in[0,1]$ is a weight parameter. We assume the same
values for the parameters of the Pareto and power function distributions
as in the previous two cases.

\subsection{Consumer Problem}

The representative household seeks to maximize its lifetime utility
by choosing consumption $C_{t}$ over time:

\begin{equation}
\max_{\{C_{t}\}}U=\int_{0}^{\infty}e^{-\rho t}u(C_{t})dt\label{eq:U}
\end{equation}

\noindent subject to the law of motion for capital:

\begin{equation}
\dot{K}_{t}=F(K_{t},L_{t};I_{t})-\delta K_{t}-C_{t}\label{eq:Kdot}
\end{equation}

\noindent for given $K_{0}$. The current-value Hamiltonian for this
problem is:

\[
H_{c}=u(C_{t})+\mu_{t}\left[F(K_{t},L_{t})-\delta K_{t}-C_{t}\right]
\]

\noindent The first-order conditions with respect to consumption and
capital are:

\begin{align*}
\frac{\partial H_{c}}{\partial C_{t}} & =u'(C_{t})-\mu_{t}=0\\
\frac{\partial H_{c}}{\partial K_{t}} & =\mu_{t}\left[F_{K}-\delta\right]=-\dot{\mu}_{t}+\rho\mu_{t}
\end{align*}
Differentiating the first optimality condition with respect to time
yields $u''(C_{t})\dot{C}_{t}=\dot{\mu}_{t}$, and substituting into
the second optimality condition gives

\begin{equation}
\frac{\dot{C}_{t}}{C_{t}}=\frac{1}{\eta(C_{t})}\left[F_{K}(K_{t},L_{t})-\rho-\delta\right]\label{eq:Euler}
\end{equation}
where $\eta(C_{t})=-\frac{u''(C_{t})C_{t}}{u'(C_{t})}$ is the elasticity
of intertemporal substitution. 

\paragraph{Limit Behavior in Region 2}

When the economy is in region 2, then $F_{K}=A$. If the agent's utility
function exhibits constant elasticity of substitution $\eta$, then
the Euler equation implies a constant growth rate of consumption 
\begin{equation}
g_{C}=\frac{\dot{C_{t}}}{C_{t}}=\frac{A-\rho-\delta}{\eta}\label{eq:gC}
\end{equation}
Let us assume that $A>\rho+\delta$ so consumption growth is positive
and consider the case that the economy remains in region 2 forever---for
example, because full automation $\Phi(I)=1$ has been reached. Then
the economy will converge towards a balanced growth path in which
$g_{Y}=g_{K}=g_{C}$ as in (\ref{eq:gC}) and the savings rate $s^{\infty}=1-C/Y$
is constant. From (\ref{eq:Kdot}), we obtain that
\[
g_{K}=\frac{\dot{K_{t}}}{K_{t}}=\frac{sA(K_{t}+L)}{K_{t}}-\delta
\]
As $\lim_{t\rightarrow\infty}L/K_{t}=0$, we can equate $g_{C}=g_{K}$
and solve for the long-run savings rate 
\[
s^{\infty}=\frac{A-\rho-\delta+\eta\delta}{A\eta}=\frac{1}{\eta}-\frac{\rho+(1-\eta)\delta}{A\eta}
\]

\paragraph*{Bounds}

Assume an initial $I_{0}$ and $K_{0}$ that satisfy $F_{K}\left(K_{0},L;I_{0}\right)\geq\rho+\delta$,
i.e., there was no excessive capital accumulation in the past. Then
the following proposition holds for any intertemporal utility function
that is linearly separable as specified in (\ref{eq:U}) with a twice
continuously differentiable, increasing, and strictly concave period
utility function $u(C)$:
\begin{prop}[Bounds for Output and Wages]
\label{prop:bounds}For any distribution $\Phi\left(i\right)$ of
tasks in compute space and exogenous growth in the automation index
$I_{t}$, the paths of capital, output, and wages lie between lower
and upper bounds $K^{-}\leq K_{t}\leq K_{t}^{+}$, $Y_{t}^{-}\leq Y_{t}\leq Y_{t}^{+}$
and $w_{t}^{-}\leq w_{t}\leq w_{t}^{+}$.

The lower bounds are defined by the fixed-capital case with $K^{-}=K_{0}$
$\forall t$ and $Y_{t}^{-}=F(K^{-},L,I_{t})$, $w_{t}^{-}=F_{L}(K^{-},L,I_{t})$.
The lower bound on wages first rises in $I_{t}$ and then declines
in $I_{t}$. It declines to $A$ in finite time if full automation
is reached asymptotically, i.e., if $\lim_{I\rightarrow\infty}\Phi(I)=1$. 

If $\Phi(I_{t})<1$, an upper bound $K_{t}^{+}$ for capital is defined
by $F_{K}\left(K_{t}^{+},L,I_{t}\right)=R=\rho+\delta$ $\forall t$
as long as a solution exists; otherwise we set $K_{t}^{+}=\infty$.
The upper bounds for output and wages are $Y_{t}^{+}=F(K_{t}^{+},L,I_{t})$
and $w_{t}^{+}=F_{L}(K_{t}^{+},L,I_{t})$. All three upper bounds
are increasing in the automation index $I_{t}$. If automation is
full, $\Phi(I_{t})=1$, the upper bounds are $K_{t}^{+}=\infty$ and
$Y_{t}^{+}=\infty$, and the upper bound on wages discontinuously
collapses to $w_{t}^{+}=A$.
\end{prop}

\begin{proof}
Observe that for any twice continuously differentiable period utility
function that is increasing and strictly concave, the elasticity in
the Euler equation (\ref{eq:Euler}) satisfies $\eta(C_{t})\in(0,\infty)$.
Consumption on the optimal path is increasing as long as $F_{K}>\rho+\delta$
and constant when $F_{K}=\rho+\delta$. Our characterization of the
factor price frontier delivers most of the remaining results. 

For the lower bound, observe that increases in $I$ and $\Phi(I)$
raise the marginal product $F_{K}$ for given $K$, triggering additional
capital accumulation, which raises output and wages above the lower
bound. For the upper bound, observe that by the Euler equation, capital
accumulation will never exceed the upper threshold $K_{t}^{+}$, which
is given by
\begin{equation}
K_{t}^{+}=\frac{A^{\sigma}L(1-\Phi(I_{t}))^{\frac{1}{\sigma-1}}\Phi(I_{t})}{(R^{\sigma-1}-A^{\sigma-1}\Phi(I_{t}))^{\frac{\sigma}{\sigma-1}}}.\label{eq:optimal_K}
\end{equation}

For given $I_{t}$, output and wages are increasing in $K_{t}$, implying
that they must lie between the lower and upper bounds defined by $K^{-}$
and $K_{t}^{+}$. If $\Phi(I_{t})=1$, the production function is
$AK$-style, and Lemma \ref{lem:automation_wages} implies that $w_{t}=A$.
\end{proof}
On the factor price frontier, the lower bound on wages $w^{-}$ is
pinned down by the automation path in Figure \ref{fig:FPF_risingAut};
it collapses to $A$ in finite time if the economy asymptotically
converges to full automation. As long as $\Phi(I)<1$, the upper bound
on wages $w_{t}^{+}$ is pinned down by the intersection of the corresponding
factor price frontier with a vertical line at $R=\rho+\delta$ and
rises without bounds in $I$. However, when full automation $\Phi(I)=1$
is reached, the upper bound on wages $w_{t}^{+}$ discontinuously
collapses to $A$, which equals the lower bound and must therefore
equal the equilibrium wage. This result is independent of intertemporal
preferences and savings behavior and occurs in finite time if the
distribution of task complexity $\Phi(I)$ is bounded, as in our two
AGI scenarios.

\medskip{}

\paragraph{The Balancing Savings Rate}

To further investigate the race between automation and capital accumulation,
we analyze the threshold at which the wage effects of automation and
capital accumulation precisely offset each other. For this, we take
the total differential of the equilibrium wage, $w_{t}=F_{L}(K_{t},L;I_{t})$,
and set $dw_{t}=0$ to find
\begin{equation}
F_{KL}(K_{t},L;I_{t})\frac{dK_{t}}{dt}+F_{LI}(K_{t},L;I_{t})\frac{dI_{t}}{dt}=0
\end{equation}
Suppose, for simplicity, that $\delta=0$ so we can denote the savings
rate at $t$ by $s_{t}=\dot{K}_{t}/Y_{t}$. Also, note that $F_{LI}(K_{t},L;I_{t})\frac{dI_{t}}{dt}=F_{L\Phi}\dot{\Phi}_{t}$.
Then 
\[
s_{t}Y_{t}\cdot F_{KL}=-F_{L\Phi}\dot{\Phi}_{t}
\]
The left-hand side is the increase in wages due to capital accumulation.
The right-hand side is the change in wages due to automation. As we
observed above in Lemma \ref{lem:automation_wages}, the term $F_{L\Phi}$
encompasses the productivity effect and the displacement effect of
automation on wages. Dividing by the cross-derivative $F_{KL}$, 
the fraction $\frac{F_{L\Phi}}{F_{KL}}$ captures the wage effects
of automation relative to capital accumulation. After some algebra,
we obtain
\[
\frac{F_{L\Phi}}{F_{KL}}=\bigg[\frac{\sigma}{1-\sigma}(k/\ell)^{\frac{1-\sigma}{\sigma}}-\bigg(\kappa+\frac{1}{1-\sigma}\bigg)\bigg]k
\]
The first term in the brackets is the productivity effect that is
increasing in the relative abundance of capital $k/\ell$ -- if capital
is very abundant compared to labor, then using capital for newly automated
tasks significantly raises output. The second term is the displacement
effect that is increasing in $\kappa$ and thus in the automation
index $I$. (Recall that for the given level of automation $I,$ $\kappa\left(I\right)=\Phi(I)/\left[1-\Phi\left(I\right)\right]$
reflects the threshold of the capital/labor ratio below which the
economy is in region 2 such that the scarcity of labor is lifted.)
Intuitively, if a large fraction of tasks has already been automated,
then further automation of marginal tasks will result in a large fall
in labor demand since the automated labor has to be reallocated to
an ever smaller set of human-only tasks. 

By rearranging terms, we obtain the following expression for the savings
rate that, given $\Phi$ and $g$, perfectly offsets the effect of
automation on wages
\begin{equation}
\tilde{s}_{t}=\bigg[\bigg(\kappa+\frac{1}{1-\sigma}\bigg)-\frac{\sigma}{1-\sigma}(k/\ell)^{\frac{1-\sigma}{\sigma}}\bigg]\cdot\frac{K_{t}}{Y_{t}}\cdot\frac{\phi_{t}I_{t}}{\Phi_{t}}\cdot g\label{eq:balanced-savings-rate}
\end{equation}
The condition tells us that, to offset the effect of automation on
wages, the savings rate must be increasing with (i) the displacement
effect net of the productivity effect, (ii) the capital-output ratio,
(iii) the relative mass of automated tasks at the current compute
threshold for task automation, and (iv) the growth of compute $g$.
Intuitively, a large fraction of output must be invested if (i) the
displacement effect reduces wages significantly, or (ii) there is
already a large amount of capital stock in the economy, or (iii) a
large amount of tasks are being automated, or (iv) automation is fast.

The expression in (\ref{eq:balanced-savings-rate}) tells us about
the threshold level for the savings rate at time $t$ above which
wages rise and below which wages fall, given the extent of automation
occurring at time $t$. In other words, it characterizes the short-run
behavior of wages as an outcome of the race between automation and
capital accumulation. 

\paragraph{Long-Run Dynamics for Unbounded Task Distributions}

To further illuminate the trade-off in (\ref{eq:balanced-savings-rate}),
we turn to the long-run dynamics of wages. To do so, we start by characterizing
the conditions for the existence of a balanced growth path (BGP).
We define a BPG as an equilibrium path on which output and capital
stock grow at a constant rate and factor shares remain constant.
\begin{lem}
\label{lem:long-run-Ramsey}Suppose that as $t$ increases, $\Phi(I_{t})\rightarrow1$,
and focus on the limit case. Then the return to capital converges
to $A$. Moreover, output and capital stock grow at the rate $(A-\rho-\delta)/\eta$
and the savings rate converges to $(A-\rho-\delta+\eta\delta)/A\eta$.
\end{lem}

\begin{proof}
If the economy is in region 1 in the limit, then the production function
converges to
\[
\lim_{\Phi\rightarrow1}A\left[K^{\frac{\sigma-1}{\sigma}}\Phi{}^{\frac{1}{\sigma}}+L^{\frac{\sigma-1}{\sigma}}(1-\Phi)^{\frac{1}{\sigma}}\right]^{\frac{\sigma}{\sigma-1}}=AK
\]
If the economy is in region 2 in the limit, then $F(K,L)=A(K+L)$.
In both cases,
\[
\lim_{\Phi\rightarrow1}F_{K}=A
\]
As a result, the Euler equation implies
\[
\frac{\dot{C}_{t}}{C_{t}}=\frac{1}{\eta}[F_{K}-\rho-\delta]\rightarrow\frac{1}{\eta}[A-\rho-\delta]
\]
as $\Phi\rightarrow1$. Output and capital must grow at the same rate,
which implies that the savings rate must satisfy 
\[
\frac{\dot{C}_{t}}{C_{t}}=\frac{\dot{K}_{t}}{K_{t}}=\frac{s^{\infty}Y_{t}}{K_{t}}-\delta
\]
Since $\lim_{\Phi\rightarrow1}\frac{Y_{t}}{K_{t}}=A$, this requires
that $\dot{K}_{t}/K_{t}\rightarrow s^{\infty}A-\delta$. Therefore,
we have
\begin{align*}
s^{\infty}A-\delta & =\frac{1}{\eta}[A-\rho-\delta]\\
s^{\infty} & =\frac{A-\rho-\delta+\eta\delta}{A\eta}
\end{align*}
\end{proof}
Since we are interested in the long-run dynamics of wages, we make
a simplifying assumption that the savings rate is given exogenously
at a constant value $s^{\infty}$, which can be interpreted as the
long-run savings rate. Under the assumption, $s$ is a key parameter
determining the rate of capital accumulation. Depending on the value
of $s^{\infty}$ relative to the rate of automation $g$, the race
between automation and capital accumlulation can result in three possible
outcomes. The following proposition summarizes the results.

\begin{figure}
\centering{}\includegraphics[scale=0.65]{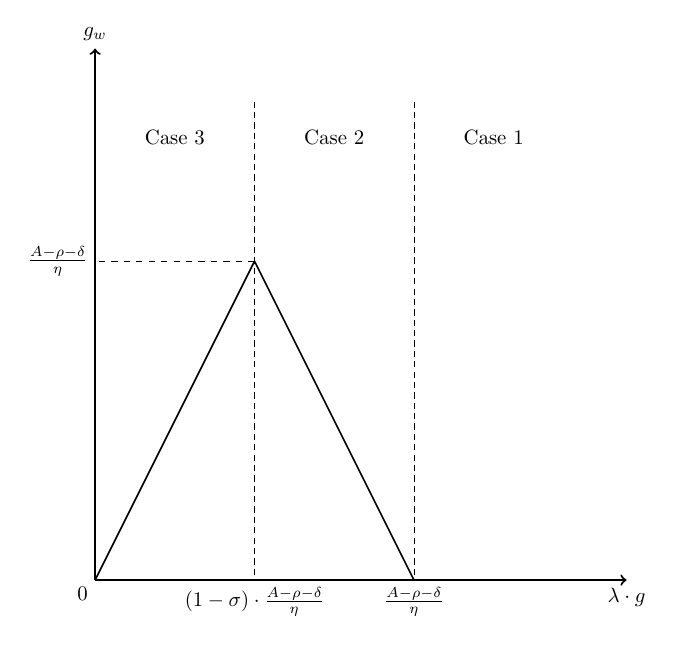}\caption{\label{fig:max_wage}Wage growth rate ($g_{w}$) as a function of
the rate of automation ($\lambda g$)}
\end{figure}

\begin{prop}[Race between automation and capital accumulation]
\label{prop:race-between-I-and-K}Suppose the complexity distribution
of tasks is Pareto and that the economy starts in region 1, i.e.,
$I_{0}<\hat{I}_{0}$. Then the growth of wages and long-run labor
shares are characterized by two thresholds on the rate of automation
$\lambda g$: 
\end{prop}

\begin{enumerate}
\item If $\lambda g>\frac{A-\rho-\delta}{\eta}$ then $\lim_{t\rightarrow\infty}w_{t}=A$
and the labor share converges to zero. 
\item If $\frac{A-\rho-\delta}{\eta}\cdot(1-\sigma)<\lambda g\leq\frac{A-\rho-\delta}{\eta}$
then wages grow exponentially at an asymptotic rate $\frac{1}{\sigma}\bigg(\frac{A-\rho-\delta}{\eta}-\lambda g\bigg)$
and the labor share converges to one. 
\item Lastly, if $\lambda g\leq\frac{A-\rho-\delta}{\eta}\cdot(1-\sigma)$
then wages grow exponentially at an asymptotic rate $\frac{\lambda g}{1-\sigma}$
and the labor share converges to $1-\bigg[\frac{(A-\rho-\delta+\eta\delta)/\eta}{\frac{\lambda g}{1-\sigma}+\delta}\bigg]^{\frac{\sigma-1}{\sigma}}$.
\end{enumerate}
\begin{proof}
See Appendix \ref{subsec:Proofs-of-Proposition-race-I-and-K}.
\end{proof}
Intuitively, the proposition illustrates how wages evolve as the result
of a race between automation and capital accumulation. As observed
above, the fraction $\frac{A-\rho-\delta}{\eta}$ is proportional
to the long-run savings rate of the economy. In the first case, if
the rate of task automation $\lambda g$ is too high compared the
savings rate, then the automation index $I$ crosses the threshold
$\hat{I}$ in finite time and the economy transitions into region
2, where wages collapse to $A$ and remain stagnant. If the rate of
task automation $\lambda g$ is at an intermediate value, then wage
growth is constrained by capital accumulation. Wages grow perpetually
at rate $\frac{1}{\sigma}(\frac{A-\rho-\delta}{\eta}-\lambda g)$,
which is proportional to the savings rate minus the rate of automation.
Finally, if $\lambda g$ is low enough, then the rate of automation
rather than capital accumulation constrains wage growth. In other
words, wage growth depends on how fast automation increases the efficiency
of factor allocation and allows the utilization of abundant capital.
Indeed, the growth rate of wages (and of the entire economy) in this
regime is increasing in the rate of automation. 

Figure \ref{fig:max_wage} illustrates the three cases in Proposition
\ref{prop:race-between-I-and-K}. The figure plots the long-run growth
rate of wages as a function of the rate of automation. If $\lambda g$
is sufficiently low as in case 1, then the wage growth rate is increasing
in the rate of automation as the upward-sloping part of the curve
indicates. Once $\lambda g$ surpasses the first threshold value,
the growth rate of wages starts to decline as $\lambda g$ increases
further. Lastly, if $\lambda g$ surpasses the second threshold value,
then wages do not grow in the long run and stay at $A$.

\subsection{Numerical Illustration}

\begin{table}
\centering{}%
\begin{tabular}{|>{\centering}m{0.08\textwidth}|>{\centering}m{0.08\textwidth}|>{\raggedright}m{0.25\textwidth}|}
\hline 
\multicolumn{1}{|c|}{Parameter} & Value & Description\tabularnewline
\hline 
\hline 
$\rho$ & 0.04 & Discount rate\tabularnewline
\hline 
$\eta$ & 2 & Risk aversion parameter\tabularnewline
\hline 
$\delta$ & 0.1 & Depreciation rate\tabularnewline
\hline 
$\sigma$ & 0.5 & Elasticity of substitution\tabularnewline
\hline 
$A$ & 0.5 & Total factor productivity\tabularnewline
\hline 
$L$ & 1 & Labor endowment\tabularnewline
\hline 
$\Phi_{0}$ & 0.608 & Initial fraction of automated tasks\tabularnewline
\hline 
$K_{0}$ & 4.6 & Initial capital stock\tabularnewline
\hline 
\end{tabular}\caption{\label{tab:calibration}Parameter values for the numerical illustration}
\end{table}
To provide an illustration of the theoretical results, we present
simulations of the four automation scenarios described in Section
3.1. Table \ref{tab:calibration} summarizes the parameter values
that were common to all the simulations. The first five parameters
are standard in the literature, and $L=1$ is a normalization. We
chose $\Phi_{0}$ and $K_{0}$ to match a 66\% initial labor share
with capital at its steady state for that level of technology.

\begin{figure}
\centering{}\subfloat[Business-as-usual scenario]{\includegraphics[width=0.48\columnwidth]{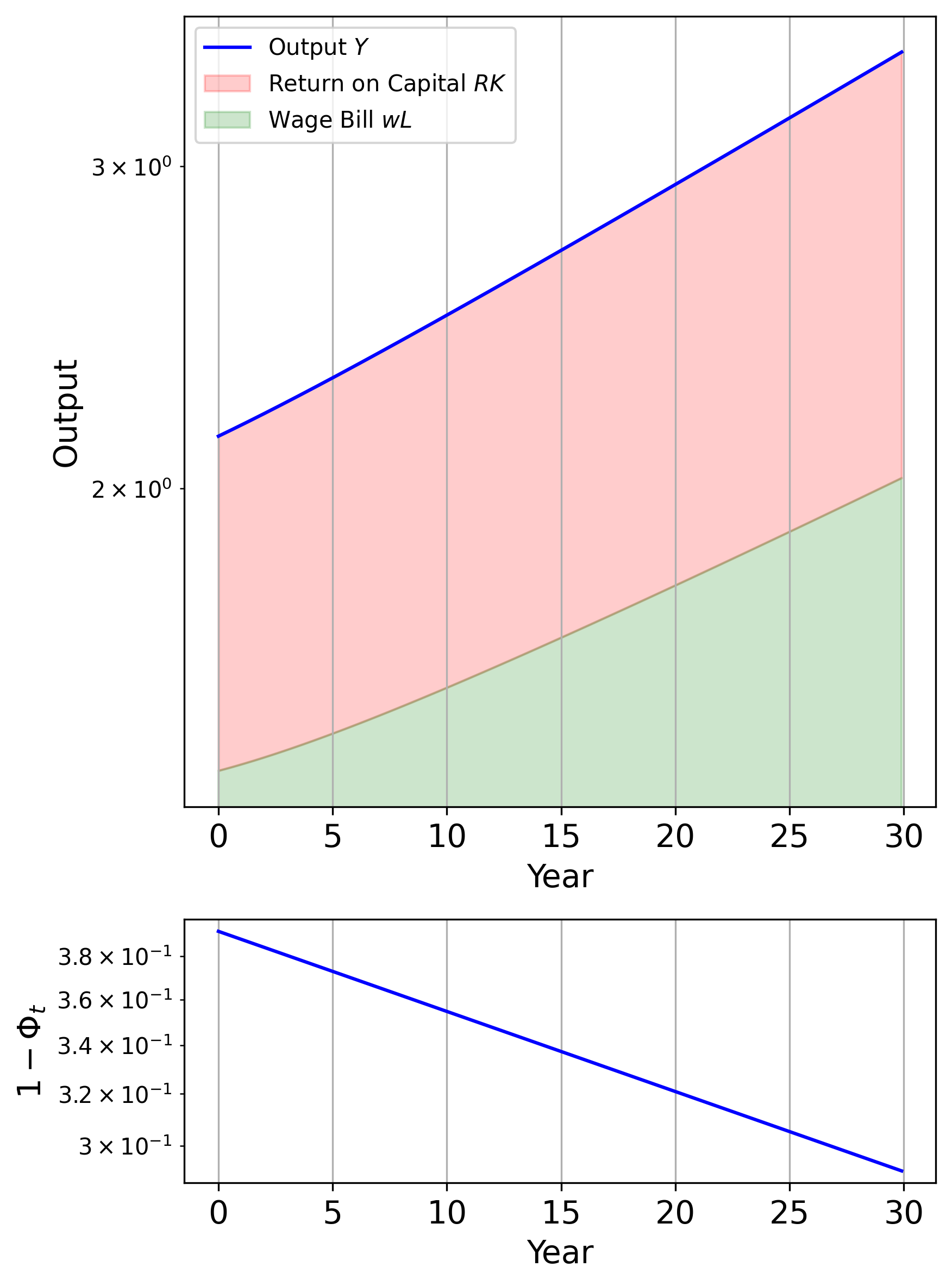}}\medskip{}
\subfloat[Baseline AGI scenario]{\includegraphics[width=0.48\columnwidth]{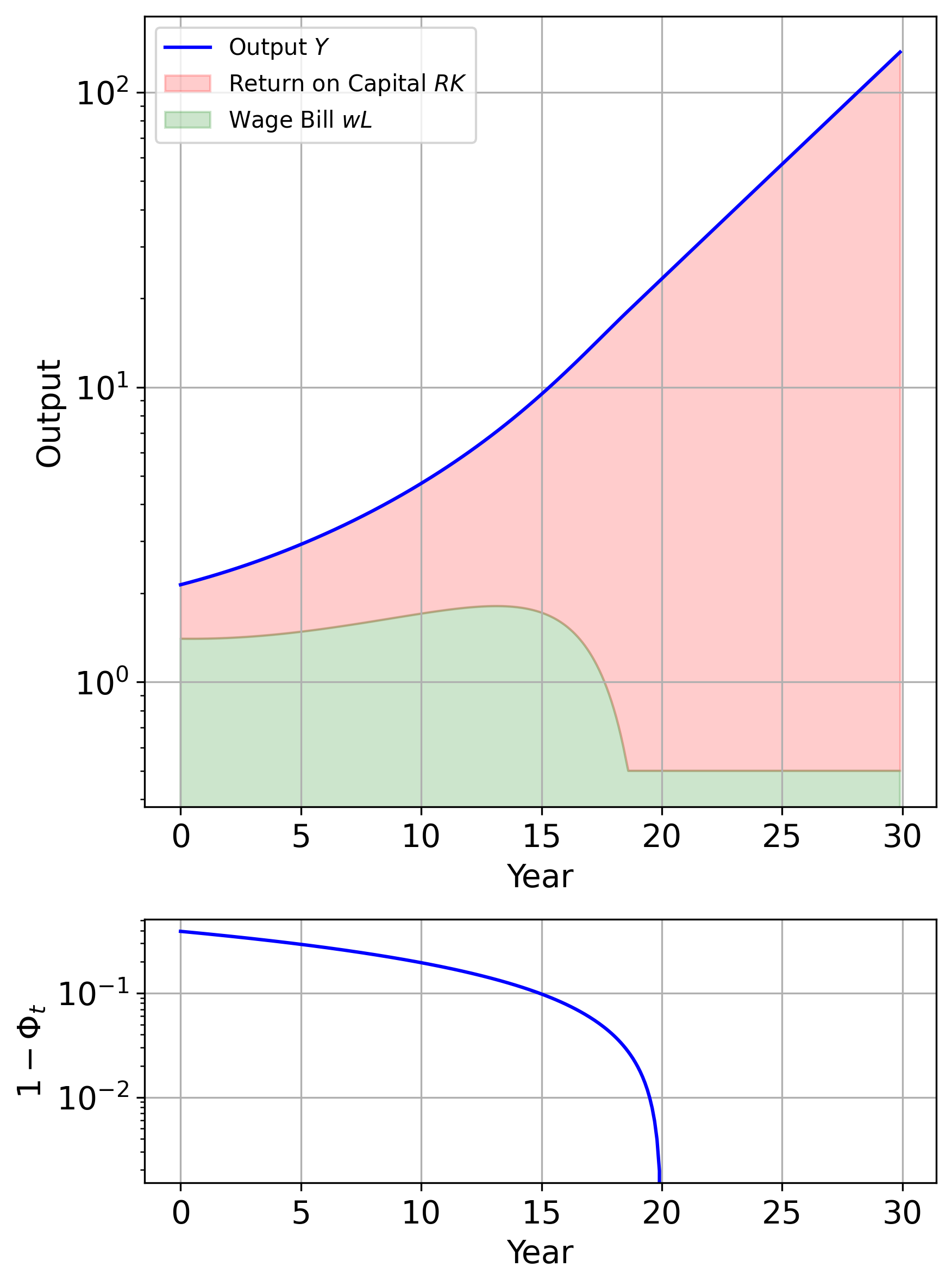}}\medskip{}
\subfloat[Aggressive AGI scenario]{\includegraphics[width=0.48\columnwidth]{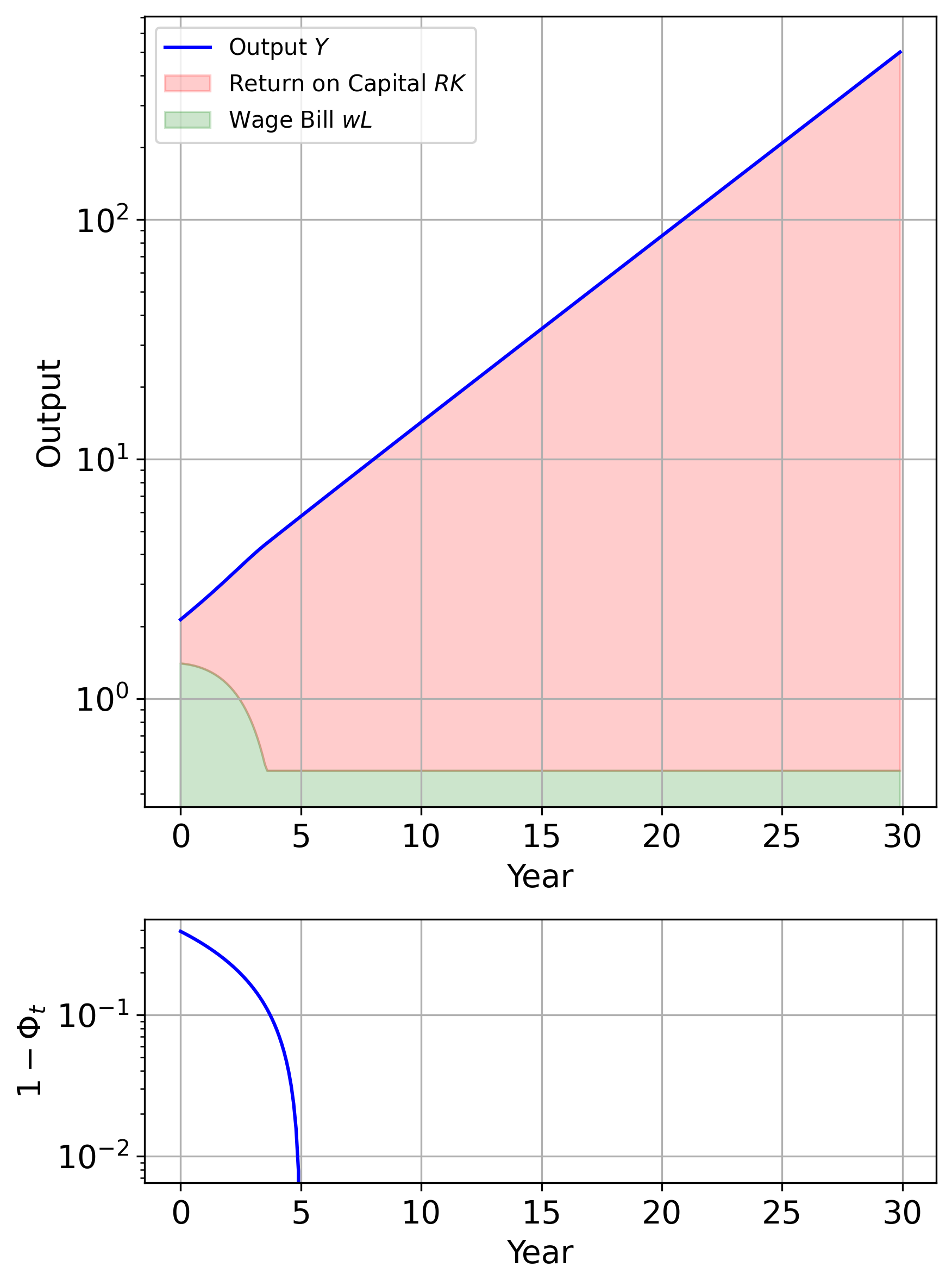}}\medskip{}
\subfloat[Mixed scenario]{\includegraphics[width=0.48\columnwidth]{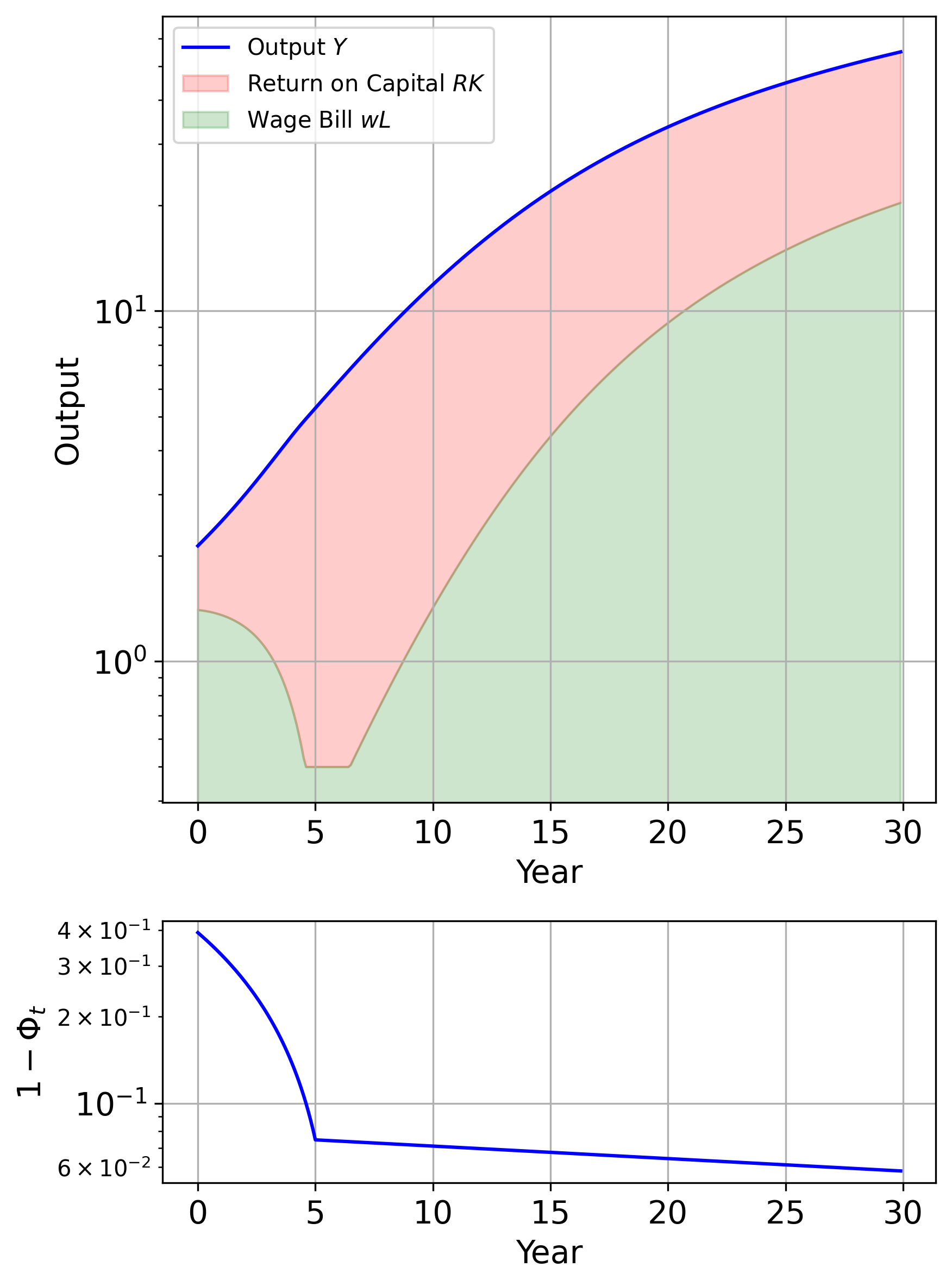}}\caption{\label{fig:scenarios}Simulations of the four scenarios}
\end{figure}
Figure \ref{fig:scenarios} presents the results. Panel (a) shows
the traditional automation scenario ``business-as-usual,'' in which
$\Phi(i)$ reflects a rate of task automation of $\lambda g=0.01$
per year. The upper part of the panel shows the output, split into
the returns to capital (red, upper area) and the wage bill (green,
lower area), on a logarithmic scale. The lower part of the panel shows
the fraction of unautomated tasks $1-\Phi$ on a logarithmic scale---for
panel (a), this is a straight line, capturing exponential decay. We
observe that in the ``business-as-usual'' scenario, output grows
at approximately 2\% per year, and both the returns to capital and
the wage bill rise approximately in tandem (with a small decline in
the labor share due to the effects of automation). Note that this
scenario corresponds to case 3 in Proposition \ref{prop:race-between-I-and-K},
i.e., capital accumulation is sufficiently fast so that growth is
constrained by the speed of automation. 

Panels (b) and (c) show the AGI scenarios, in which the fraction
of unautomated tasks collapses to zero in 20 or 5 years, respectively.
In the baseline AGI scenario, wages grow slightly during the initial
periods but then collapse before full automation is reached. After
the collapse, wages are equal to the returns to capital, and the economy
remains in region 2 where labor and capital are perfectly substitutable,
with steady-state growth of 18\% per year. In the aggressive AGI scenario,
the wage collapse happens after about 3 years. Since the scarcity
of labor is relieved earlier than in the baseline AGI scenario, the
growth take-off occurs earlier. 

Panel (d) shows the ``bout-of-automation'' scenario. During the
initial periods, a large fraction of tasks are automated, leading
to wage collapse similar to the aggressive AGI scenario as the economy
enters region 2---labor is abundant because of the rapid automation
and comparatively low capital stock. However, over time, the economy
accumulates more capital, making labor scarcer again. Around year
9, the economy has accumulated sufficient capital so that it returns
to region 1. Wages rise above $A$ and start growing again in line
with further (slower) advances in automation and further capital accumulation.
This scenario illustrates the possibility that labor demand may collapse
due to rapid automation but recover later because of a long tail of
tasks that cannot be automated.

\section{Extensions}

\subsection{\label{ssec:fixedfactors}Fixed Factors and the Return of Scarcity}

If labor is dethroned as the most important factor of production,
it becomes useful to disentangle the remaining factors, which have
traditionally been lumped together into ``capital'' in the economic
models of the Industrial Age. Let us distinguish between factors that
are in fixed supply and factors that are reproducible and can therefore
be accumulated. We continue to call all reproducible factors ``capital,''
including compute, robots, power plants, and factories. By contrast,
factors in fixed supply include land, space, minerals, or solar radiation.\footnote{During the Industrial Age, labor was considered in fixed supply at
the relevant times scale -- raising humans took so long that their
supply could be approximated as exogenous -- whereas human capital
was a reproducible factor.} It is difficult to predict which scarce factors will matter the
most in an AGI-powered future -- in the short term, it is likely
that microchips and the semiconductor fabrication equipment (``fabs'')
used for producing these chips will be bottlenecks, but these are
clearly reproducible. By contrast, the raw materials going into the
production of chips, for example certain rare earth minerals, are
irreproducible. In the longer-term, matter or, equivalently, energy
($E=mc^{2}$) may be the ultimately source of scarcity.

For the purposes of our analysis, we incorporate a fixed factor in
our analysis that we label $M$ for minerals or matter. We assume
that the aggregate production function is a Cobb-Douglas aggregator
of the task composite and $M$,

\begin{equation}
Y=A\left[\int_{i}y(i){}^{\frac{\sigma-1}{\sigma}}d\Phi(i)\right]^{\frac{\sigma}{\sigma-1}\cdot\alpha}M^{1-\alpha}\label{eq:production-ftn-land}
\end{equation}
where $\alpha\in[0,1]$ is the share of the composite among total
output. Then, a version of Lemma \ref{lem:scarcity-of-labor} applies,
separating two regimes:
\begin{lem}
\label{lem:scarcity-of-labor-land}For given $(K,L)$, the automation
threshold $\hat{I}$ is defined by (\ref{eq:Ihat}) as in the original
lemma and is independent of $M$. It defines two regions:

\noindent \textbf{Region 1:} If $I<\hat{I}$, then labor is scarce
compared to capital and employed only for unautomated tasks. Output
is given by

\begin{align}
Y & =F\left(K,L,M;I\right)=A\left[K^{\frac{\sigma-1}{\sigma}}\Phi(I)^{\frac{1}{\sigma}}+L{}^{\frac{\sigma-1}{\sigma}}(1-\Phi(I))^{\frac{1}{\sigma}}\right]^{\frac{\sigma}{\sigma-1}\cdot\alpha}M^{1-\alpha}\label{eq:Yeff-CD-case1}
\end{align}
Wages and the returns to $M$ satisfy
\begin{align*}
w & =\alpha A\left[K^{\frac{\sigma-1}{\sigma}}\Phi(I)^{\frac{1}{\sigma}}+L{}^{\frac{\sigma-1}{\sigma}}(1-\Phi(I))^{\frac{1}{\sigma}}\right]^{\frac{\sigma}{\sigma-1}\cdot\alpha-1}L^{-\frac{1}{\sigma}}(1-\Phi(I))^{\frac{1}{\sigma}}M^{1-\alpha}>R\\
Q & =(1-\alpha)Y/M
\end{align*}

\noindent \textbf{Region 2:} If $I\geq\hat{I}$, then the relative
scarcity of labor is relieved; if the inequality is strict, labor
and capital are perfect substitutes for the marginal task. Output
is given by 
\begin{align}
Y & =F\left(K,L,M\right)=A(K+L)^{\alpha}M^{1-\alpha}\label{eq:Yeff-CD-case2}
\end{align}
Wages and the return to $M$ satisfy
\begin{align}
w & =R=\alpha A(K+L)^{\alpha-1}M^{1-\alpha}\label{eq:w_with_M_reg2}\\
Q & =(1-\alpha)A(K+L)^{\alpha}M^{-\alpha}\nonumber 
\end{align}
\end{lem}

\begin{proof}
The proof follows along the same lines as the proof of Lemma \ref{lem:scarcity-of-labor}.
\end{proof}
The presence of the fixed factor $M$ does not affect the key characteristics
of the production function described in Lemma \ref{lem:scarcity-of-labor}
such as the threshold for the automation index beyond which labor
is no longer scarce compared to capital. Similar results apply for
the effects of automation on wages:
\begin{lem}[Automation and Wages with $M$]
For given capital intensity $K/L$, an increase in automation $d\Phi(I)$
always raises $R$ for $I<\hat{I}$. The effects on $w$ is hump-shaped:
there is a threshold $I^{\ast}(K/L)$ with $\Phi(I^{\ast}(\cdot))\in(0,1)$
such that wages $w$ rise in $\Phi(I)$ as long as $I<I^{\ast}(K/L)$
but decline in $\Phi(I)$ for $I>I^{\ast}(K/L)$. The threshold $I^{\ast}$
with $M$ is lower than in Lemma \ref{lem:automation_wages}. In the
limit cases of $\Phi(I)=0$ and $\Phi(I)\geq\kappa/(1+\kappa)$, wages
are given by (\ref{eq:w_with_M_reg2}). The limit is reached for any
$K/L$ ratio if $\Phi(I)=1$.
\end{lem}

\begin{proof}
The effect of automation on wages for a given $K/L$-ratio is similar
to Lemma \ref{lem:automation-and-wages}: 
\begin{align}
\frac{d\log w}{d\Phi(I)}= & \bigg(\frac{\sigma}{\sigma-1}\alpha-1\bigg)\frac{1}{\sigma}\bigg(k^{\frac{\sigma-1}{\sigma}}-\ell^{\frac{\sigma-1}{\sigma}}\bigg)(Y/A)^{\frac{1-\sigma}{\sigma}}-\frac{1}{\sigma}\frac{1}{1-\Phi(I)}\label{eq:automation-wages-1}
\end{align}
The only difference is the multiplicative term $\frac{\sigma}{\sigma-1}\alpha-1$,
which is smaller than $\frac{1}{\sigma-1}$ for $\alpha<1$. Thus,
the productivity effect is smaller with the fixed factor $M$, meaning
that wages start to decline at lower levels of $I$.
\end{proof}
Although the presence of a fixed factor preserves the key results
on the automation threshold and the wage effects of automation, the
long-run dynamics of the economy change---for the worse. In particular,
we find that wages will always decline to the return on capital as
the economy will always enter region 2 in finite time.
\begin{prop}
\label{prop:race-betwee-I-and-K-M}If $\lim\Phi(I)=1$, then the economy
enters region 2 in finite time, and wages equal the returns on capital
$w=R=\rho+\delta$. The labor share equals $\alpha L/(K^{\ast}+L)$,
where $K^{\ast}$ is defined by
\[
w=R=\alpha A(K^{\ast}+L)^{\alpha-1}M^{1-\alpha}=\rho+\delta
\]
\end{prop}

\begin{proof}
If the economy is in region 2 after some finite time $T$, it will
converge towards a steady state in which $(K^{\ast}+L)$ are pinned
down by the Euler equation (\ref{eq:Euler}), resulting in the expression
above. We observe that $K^{\ast}$ is the maximum capital level that
an optimizing agent will accumulate in this economy since $F_{K}<\rho+\delta$
for any region 1 allocation, as can be seen from the economy's factor
price frontier. This implies that the economy will enter region 2
no later than when the automation threshold reaches the scarcity of
labor threshold $\hat{I}$ s.t.~$\Phi(\hat{I})=K^{\ast}/L/(1+K^{\ast}/L)$,
as defined in Lemma 1.

\end{proof}
Intuitively, Proposition \ref{prop:race-betwee-I-and-K-M} tells us
that if there is a fixed factor then automation eventually outpaces
capital accumulation regardless of the distribution of tasks. This
contrasts with Proposition \ref{prop:race-between-I-and-K}, which
shows that wages may grow indefinitely if a sufficient amount of tasks
is always left to labor. 

\begin{figure}

\centering{}\includegraphics[width=0.56\textwidth]{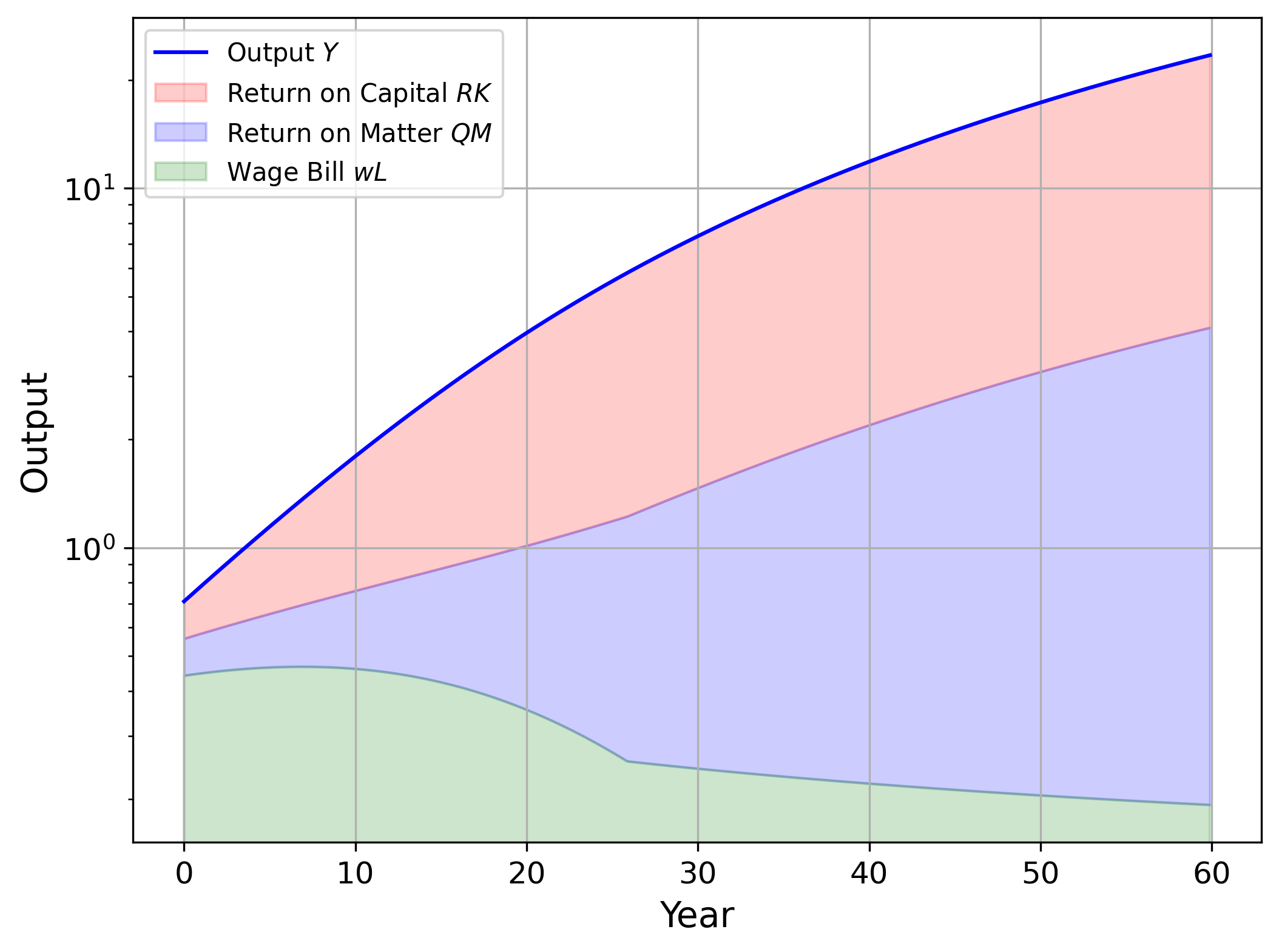}\caption{\label{fig:simulation_fixed}Factor shares with fixed factor $M$
in traditional scenario}
\end{figure}
Figure \ref{fig:simulation_fixed} illustrates the implications under
the assumption that the Cobb-Douglas for $M$ is $(1-\alpha)=.10$
in the ``traditional scenario,'' in which a constant fraction of
tasks is automated every period. As can be seen, wages peak after
about 10 years, and the economy enters region 2 after 25 years, slowly
converging to the steady-state level of capital $K^{\ast}$. In stark
contrast to our simulation results in Section 3, this illustrates
that even though there is an infinite tail of unautomated tasks, the
presence of a fixed factor bottlenecks capital accumulation and implies
that labor loses its scarcity status in finite time.

\subsection{Automating Technological Progress}

Our analysis so far has focused on automation as the only form of
technological advancement and has taken as given the technology parameter
$A$, which is considered as the main driver of productivity gains
in the neoclassical growth model. This has allowed us to derive a
number of powerful results on the effects of automation in goods production
on output and wages. However, there are widespread predictions that
advances in AI not only will make output production more efficient
but also will speed up technological progress \citep{aghional19,agrawalal2023scientific,davidson23}.

At the most basic level, the production of R\&D that drives technological
progress consists of atomistic computational tasks---like any other
production process described earlier in the paper. For example, \citet{agrawalal2023scientific}
suggest that scientific hypothesis generation can be viewed as the
making of predictions over a vast combinatorial space. We denote the
complexity distribution of tasks involved in R\&D by the distribution
function ``Gamma'' $\Gamma(i)$, which may differ from the complexity
distribution of tasks $\Phi(i)$ involved in producing output---perhaps
R\&D involves on average more complex computational tasks. W.l.o.g.,
we assume that our ability to automate both R\&D and production are
captured by the same automation index $I$.

Building on our earlier task production function and on the endogenous
growth setup of \citet{jones1995}, we assume that advances in the
technology parameter $A$ are driven by an ideas production function
combining atomistic computational tasks $\{x(i)\}$ that involve computational
complexity as reflected in the distribution function $\Gamma(i)$,
\[
\log\dot{A}=\log A^{\theta}+\int\log x(i)d\Gamma(i),
\]
where the parameter $\theta$ captures the potential for knowledge
spillovers or for decreasing returns to knowledge accumulation. Similarly
with the production of final goods, we assume that automated tasks
can be performed by either capital or labor whereas unautomated tasks
require labor,
\[
x\left(i\right)=\begin{cases}
k_{A}\left(i\right)+\ell_{A}\left(i\right) & \text{for }i<I,\\
\ell_{A}\left(i\right) & \text{for }i\geq I.
\end{cases}
\]

To keep our analysis tractable, we assume that there is an exogenous
supply of knowledge workers $L_{A}=1$ that can only work in ideas
production in addition to the unit supply of workers $L_{Y}=1$ who
are solely engaged in final output production. Moreover, we assume
unitary elasticity of substitution between tasks in the output production
function so $\sigma=1$. Analogs of lemma 1 hold for both production
functions. As long as labor is scarce (region 1), we observe that
the production functions of final output and knowledge satisfy $F(K,L)\simeq AK_{Y}^{\Phi(I)}L_{Y}^{1-\Phi(I)}$
and $\dot{A}\simeq A^{\theta}K_{A}^{\Gamma(I)}L_{A}^{1-\Gamma(I)}$
where $K_{Y}$ and $K_{A}$ are the aggregate amounts of capital employed
in final output or knowledge production.

Following the hypothesis of \citet{aghional19} and the proof of \citet{trammellk23},
it can then be shown that once automation in the two production functions
has proceeded sufficiently, growth in the described economy will experience
what \citet{aghional19} term a ``type II singularity.'' The intuition
is that a rapidly growing capital stock generates an explosion in
R\&D output and technological progress that feeds on itself, resulting
infinite output in finite time. The following proposition states this
result formally under the assumptions of a constant savings rate $s\in(0,1)$,
a constant allocation of capital across final output and ideas production,
and no depreciation for tractability.
\begin{prop}
The economy enters a path of super-exponential growth in technology
$A$ and output $Y$ that diverges to infinity in finite time once
the automation index $I$ reaches the level $I^{\heartsuit}$such
that 
\[
\frac{\Gamma\left(I^{\heartsuit}\right)}{\left(1-\Phi\left(I^{\heartsuit}\right)\right)\left(1-\theta\right)}>1
\]
The marginal product of labor in the production of final output diverges
to infinity alongside output.
\end{prop}

\begin{proof}
Observe that once the described threshold has been passed, the production
functions $F(K,L)\simeq AK_{Y}^{\Phi(I^{\heartsuit})}L_{Y}^{1-\Phi(I^{\heartsuit})}$
and $\dot{A}\simeq A^{\theta}K_{A}^{\Gamma(I^{\heartsuit})}L_{A}^{1-\Gamma(I^{\heartsuit})}$
for constant fractions of capital $K_{Y}=cK$ and $K_{A}=(1-c)K$
are lower bounds for the actual production functions using the (still
increasing) automation level $\Phi(I)$ and $\Gamma(I)$. The result
is then a direct application of the proof in \citet{trammellk23}
(Appendix B.2). Note that $w_{t}\geq A_{t}$ no matter if output production
is in Region 1 or Region 2 as defined in Lemma 1. Therefore the marginal
product of labor in the production of final output also diverges to
infinity.
\end{proof}
The condition in the proposition depends on three parameters -- sufficient
automation in the production of ideas $\Gamma(\cdot)$, sufficient
automation in the production of final output $\Phi(\cdot)$, and sufficient
returns to the accumulation of knowledge. Remarkably, as long as the
production of final output has been sufficiently automated (sufficiently
high $\Phi$), the remaining two parameters can take on any finite
levels. This highlights that sufficient automation in the production
of final output and thus capital accumulation will always lead to
an explosion in growth.

\begin{figure}
\centering{}\subfloat{\includegraphics[width=0.56\columnwidth]{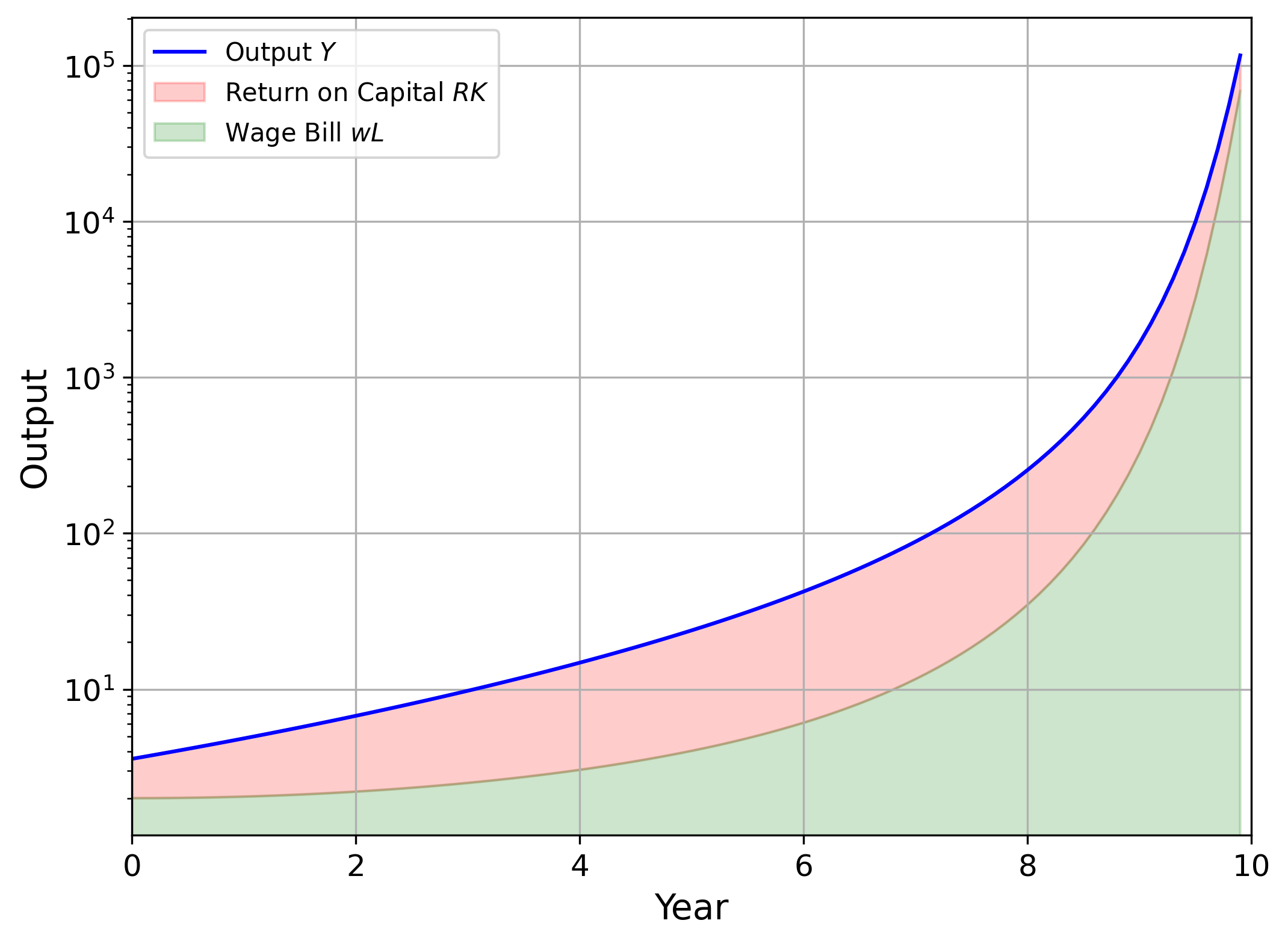}}\caption{\label{fig:tech-progress}Output and wage growth under technological
progress}
\end{figure}

Figure \ref{fig:tech-progress} illustrates the path of wages under
a CES production function for output and optimal savings, numerically
illustrating the explosive path of output growth and type-II singularity
that we identified in the proposition. The convexity of the curve
on a log-scale indicates that the growth rate of wages is ever-increasing
due to the acceleration of technological progress. (Note that the
log scale hides that the labor share of output is declining.) In the
figure, we have cut off the simulation at $t=10$. The singularity
occurs shortly thereafter. \medskip{}

In summary, even if automation induces wages to collapse to the returns
to capital at $A$, rapid technological progress from the automation
of R\&D allows workers to benefit from the advancement of AI once
sufficient automation has taken place. 

More generally, the force described in this subsection is plausible,
and sufficient progress in AI will likely indeed lead to rapid technological
advances and increases in living standards. At the same time, it is
also likely that both the production of output and of ideas will eventually
be bottlenecked by fixed factors, as we emphasized in Section \ref{ssec:fixedfactors}.
A model that comprehensively incorporates both effects is beyond the
scope of this paper.

\subsection{Nostalgic Jobs or Limits on Automation}

Our baseline model assumed that the automation of work was driven
solely by technological factors, occurring as soon as the compute
requirements of performing specific tasks were reached. However, even
if it is technologically feasible to perform certain tasks, our society
may decide that it is preferably for those tasks to remain exclusively
human. For example, \citet{korinekj23} observe that jobs such as
priests, judges, or lawmakers may remain exclusively human long after
the time when they can be performed at equal or superior levels by
machines, labelling such jobs ``nostalgic jobs.'' 

For the purposes of our analysis, we assume that there is a separate
distribution function $\Psi(I)$ that captures how far the automation
index $I$ must advance for society to choose to automate task $I$
-- in addition to the distribution $\Phi\left(I\right)$ capturing
the technological possibility of automation. The inequality $\Psi\left(I\right)\leq\text{\ensuremath{\Phi(I)}}$
reflects that society can only choose to automate tasks that are feasible
to automate. The inequality is strict if there are tasks that could
be automated from a technical perspective but aren't for societal
reasons. If $\lim_{I\rightarrow\infty}\Psi(I)<\Phi(I)\leq1$, then
this captures that there are tasks that humans choose to never automate
even though they could be. 

The described setup can also capture situations in which tasks are
delegated to machines with a delay, i.e., for higher levels of the
automation index $I$ than what is technologically feasible. \citet{korinekj23}
describe two reasons for why this may occur: First, as the capabilities
of machines to perform certain tasks become better and better than
human abilities, it may become increasingly untenable for the tasks
to be left to humans. For example, if AI systems demonstrably become
much fairer judges with fewer biases and noise than human judges,
it may become untenable to leave many judicial deicisons to error-prone
humans. Second, with sufficient advances in robotics, it may become
more and more difficult to distinguish humans and AI-powered robots
performing human services. They observe that a robot priest with greater
emotional intelligence than humans and a more comprehensive theory
of human minds than a human priest may be able to perform the tasks
typically performed by human priests quite perfectly, or intentionally
somewhat imperfectly so as to not give away that it is a robot. Both
of these categories require that the performance of AI systems is
sufficiently above human levels, corresponding to a sufficiently high
level of the automation index $I$.

\paragraph{Maximizing Wage Growth}

Consider the problem of a government with the objective to maximize
wage growth by imposing limits on automation and choosing an optimal
path $\Psi(I)\leq\Psi(I)$. The following result characterizes the
optimal $\Psi(I)$ among all Pareto distributions---given exponential
advances in the automation index $I$, this amounts to the government
choosing an optimal constant rate of automation per time period.
\begin{prop}[Maximizing Wage Growth]
Suppose $\Psi$ is a Pareto distribution defined as $\Psi(I_{t})=1-I_{0}^{-\lambda}e^{-\lambda gt}$
where $I_{0}$ is the initial automation index and $\lambda g$ is
the rate of task automation. Then the long-run growth rate of wages
is maximized for $\lambda g=(1-\sigma)\cdot\frac{A-\rho-\delta}{\eta}$
, assuming that $\Psi(I)\leq\Psi(I)\forall I$ for this distribution.
As a result, wages grow at rate $\frac{A-\rho-\delta}{\eta}$. 
\end{prop}

\begin{proof}
The proof follows from Proposition \ref{prop:race-between-I-and-K}.
The rate of automation is lowest in case 3. And the wage growth rate
is increasing in $\lambda g$. Thus, the wage growth rate increases
until $\lambda g=(1-\sigma)\cdot\frac{A-\rho-\delta}{\eta}$. Once
$\lambda g$ surpasses $(1-\sigma)\cdot\frac{A-\rho-\delta}{\eta}$,
the growth rate of wages decreases in $\lambda g$ until $\lambda g=\frac{A-\rho-\delta}{\eta}$
at which the growth rate equals zero. Therefore, the maxmum growth
rate of wages is $\frac{A-\rho-\delta}{\eta}$ at $\lambda g=(1-\sigma)\cdot\frac{A-\rho-\delta}{\eta}$.
Figure \vref{fig:max_wage} provides a graphical illustration of this
finding---the peak of the wage growth rate as a function of $\lambda g$
is $\frac{A-\rho-\delta}{\eta}$. 
\end{proof}
Figure \ref{fig:limits_Y_w} shows what happens if we slow down progress
in the ``baseline AGI scenario'' from Section 3 so that wages growth
is maximized. Up until period 14, a wage-maximizing planner is constrained
by the natural pace of automation and sets $\Psi(I)=\Phi(I)$ over
that stretch. After that point, the baseline AGI scenario implies
rapid declines in the labor share, but the planner sets $\Psi(I)<\Phi(I)$
to slow down effective automation. The left panel of the figure shows
the paths of output and wages, and the right panel depicts the two
variables in relative terms for the two scenarios. Up until period
14, the paths in the two scenarios roughly coincide (with a minor
gap opening since the AGI scenario triggers rapid capital accumulation
in advance of the economy achieving full automation). Thereafter,
the wage-maximizing planner obtains a path of exponentially growing
wages, as predicted by the proposition, whereas wages in the AGI scenario
collapse. Notably, the right-hand panel also illustrates the output
cost of foregoing the possibility of full automation. As can be seen,
the output cost of holding back automation is low at first, but eventually,
almost 100\% of the output potential of the economy is lost by holding
back automation.
\begin{figure}
\centering{}\includegraphics[width=0.48\columnwidth]{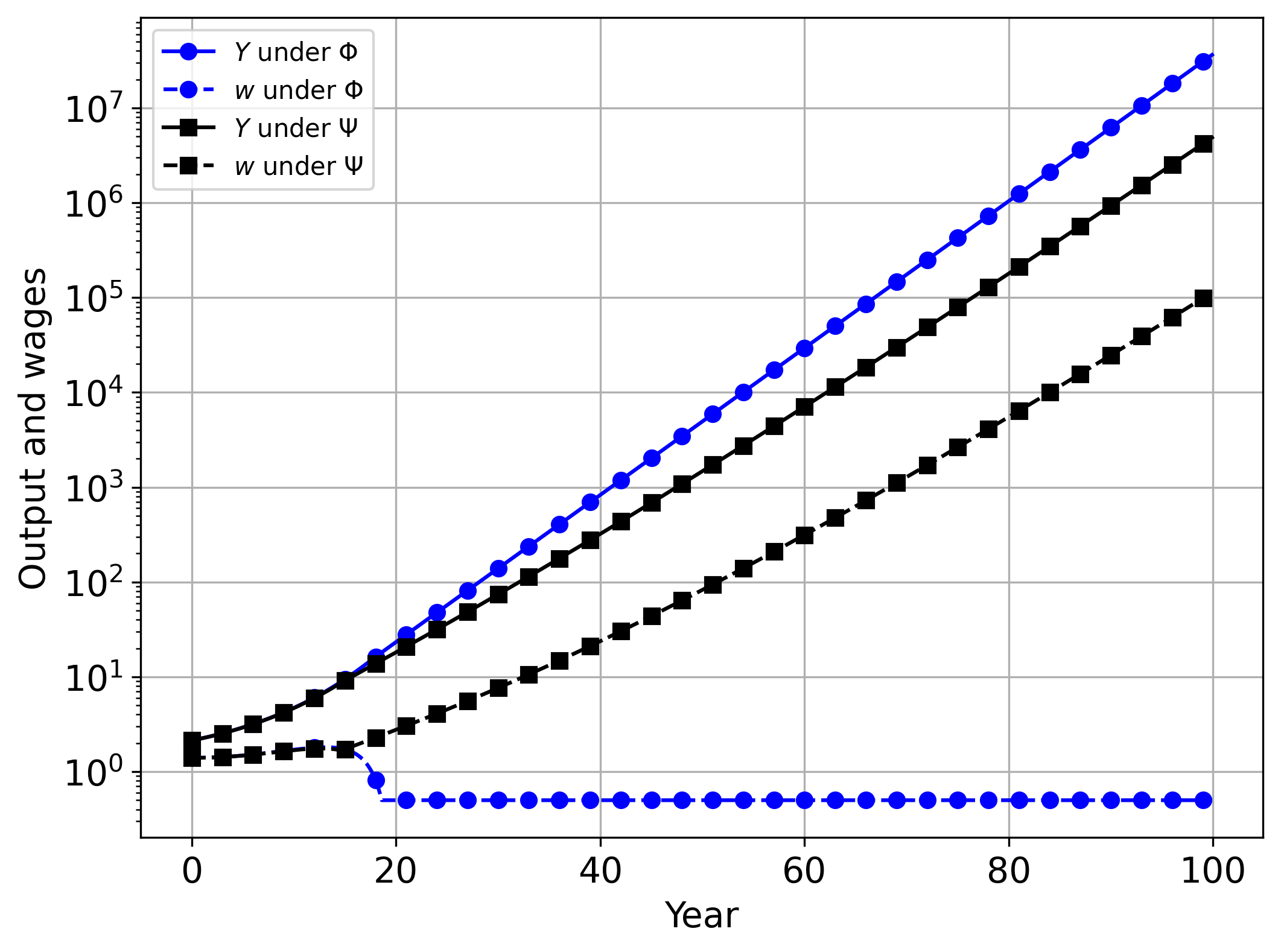}\includegraphics[width=0.48\columnwidth]{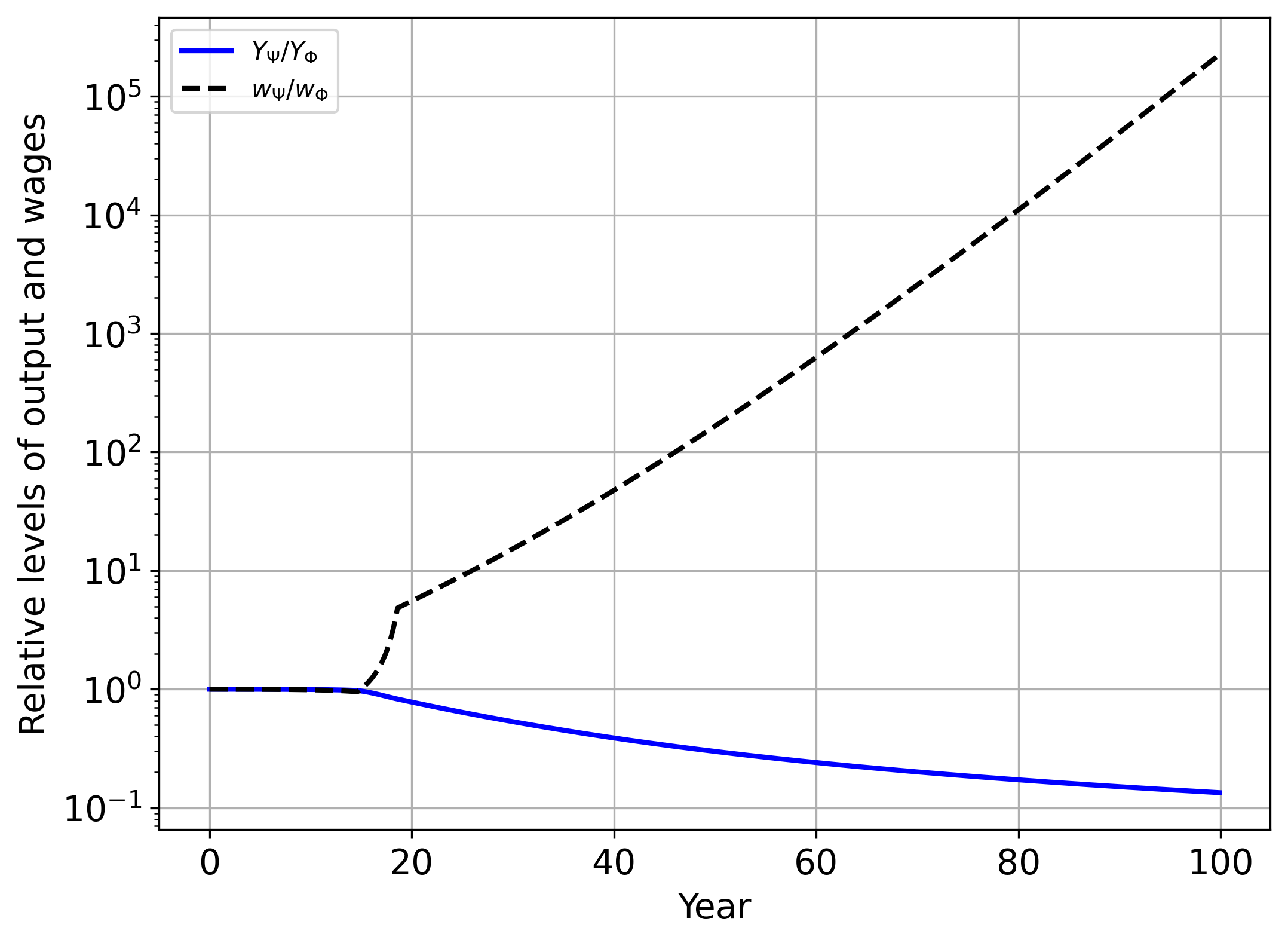}\caption{\label{fig:limits_Y_w}Comparison of output and wages under $\Phi$
and $\Psi$}
\end{figure}

Our finding illustrates that slowing down automation may be a powerful
tool to increase wages, albeit it comes at the cost of reducing output
growth. The described policy is feasible under both of the AGI scenarios
simulated in the previous section and always results in exponentially
growing wages instead of the collapse wihin a matter of years that
would otherwise occur when AGI automates human tasks too quickly.

\subsection{Heterogeneous Worker Skills}

When labor is heterogeneous, individuals are hit by the effects of
automation at different times, depending on the extent to which their
skills are automated. In practice, workers differ along many different
dimensions, and each worker's labor may be complemented or substituted
for in different ways by technological advances. One of the classical
ways of accounting for heterogeneity in the labor market, going back
to \citet{katzm1992changes}, is to split workers into skilled and
unskilled based on a threshold level of educational attainment. An
additional distinction, introduced by \citet{autoral2003}, was to
categorize workers according to whether they hold cognitive or manual
jobs performing routine or non-routine activities. Under the described
paradigm, we could capture the distribution of tasks in compute space
separately for each of the resulting buckets (e.g., routine cognitive
workers), and analyze how advances in computing capabilities will
affect that type of workers. Recent advances in AI have raised the
possibility that many cognitive tasks, including non-routing tasks,
may be automated relatively soon \citep[e.g.][]{korinek23llm_wp}.
However, ongoing advances in robotics make it likely that non-routine
manual jobs will be similarly affected to cognitive tasks by the recent
wave of progress in foundation models \citep{ahn2022robotic}.

For our purposes here, we found it useful to consider labor that differs
in uni-dimensional but continuous manner. We assume that workers differ
in an exogenous parameter that we label skill $J$, which reflects
the maximum level of task complexity that the worker can perform.
Workers' skill levels are described by the distribution function $\Upsilon(J)$.
For analytical simplicity, we assume that $\Phi(I)\geq\Upsilon(I)$.\footnote{This assumption implies that for any level of automation $I$, there
are sufficiently many skilled workers at all unautomated complexity
levels left so that we can treat unautomated workers as perfect substitutes
for each other.}

For a level of the automation index $I$, a fraction $\Upsilon(I)$
of workers are perfectly substitutable by machines and earn wage $w_{j}=A$.
A fraction $1-\Upsilon(I)$ is not substitutable, but given that the
remaining workers are sufficiently skilled, they are all effective
substitutes for each other and earn wage $w_{j}=F_{L}\left(K+\Upsilon(I),1-\Upsilon(I)\right)$.
In contrast to our baseline model, this captures the concern that
automation may make workers on the lower rungs of the skill distribution
redundant, whereas workers who are able to perform at higher levels
of skill may benefit from automation. 

In the long run, assuming less than full automation ($\Phi(I)<1$
for any finite $I$), the share of workers who are gainfully employed
will decline over time and will asymptote to $1-\lim_{J\rightarrow\infty}\Upsilon(J)$,
i.e., only workers who can perform unautomated tasks with arbitrary
computational complexity will earn higher returns than capital. If
$\Upsilon(J)=1$ for finite $J$, then the role of all human labor
will lose its scarcity value in the same manner as in the AGI scenarios
in our baseline model. Conversely, if $\Upsilon(J)$ asymptotes to
1, then there may be ever-growing inequality among workers: an ever-declining
fraction of workers at the top may see incomes rise without bounds,
whereas a fraction of the population that asymptotes towards one will
see wages collapse to the level that equates the return on capital
$A$.

\paragraph{Heterogeneity in both skill and productivity}

The described setup could easily be extended to include heterogeneity
in individual worker productivity in addition to heterogeneity in
skill. Assume that workers not only have different skill levels $J_{j}$
but are also endowed with different efficiency units of labor $L_{j}$
per time period. This may capture, for example, that there may be
two economists who can both write papers up to complexity $J$, but
one of them is twice as fast at it than the other. This could explain
the empirial observation that workers in the same occupation sometimes
earn significantly different wages.

\paragraph{Complementary human capital}

An alternative lens that may be relevant in the current era of cognitive
automation is that workers possess different levels of human capital
that is affected by automation. To keep our discussion simple, assume
again that each worker $j$ is characterized by a skill level $J_{j}$
as well as an exogenous amount of human capital $H_{j}>0$, which
enables them to supply $L_{j}=H_{j}$ efficiency units of labor per
time period. As the automation index $I$ surpasses a given worker
with skill level $J_{j}$, the human capital that they possessed is
fully devalued. The loss is greater and more painful for workers with
more human capital.

\subsection{Compute as Specific Capital}

An important feature of the ongoing AI take-off is the scarcity of
compute. In our baseline model, we followed the standard neoclassical
practice of modeling capital as uniform, capable of being deployed
in the production of any task. As illustrated in Figure \ref{fig:water_levels},
automation of new tasks then implies that the existing capital stock
can be smoothly allocated to a larger number of tasks, unlocking immediate
productivity gains. 

In practice, however, many types of capital are specific to the task
for which they were created and difficult or impossible to reallocate,
corresponding to what the literature has traditionally called putty-clay
capital.\footnote{The notion was first introduced by Leif Johansen (1959) -- once putty
has been turned into clay, it cannot be turned into another shape
-- and expanded by Solow (1962) and others.} In the current context, the most salient type of specific capital
on which AI systems rely is compute, which is in very limited supply,
slowing down the deployment of AI systems for new tasks. Another example
of specific capital is organizational capital, including the capital
derived from investments into developing new processes for deploying
new technologies in firms.

We expand our framework by assuming that each unit of capital investment
is specific to a task $i$ and can only be invested once the task
is automated, i.e., once $I\geq i$. This leaves the task production
function (\ref{eq:y(i)}) unaffected but modifies the capital accumulation
constraint: instead of a single law of motion for capital (\ref{eq:Kdot}),
the consumer needs to separately keep track of each type of capital
$k\left(i\right)$ since capital that is deployed for one task cannot
be redeployed later. In an economy in which automation is proceeding
slowly and steadily, the consumer problem is unchanged as the resulting
constraints on capital redeployment are slack---every instant of
time, a density $\phi(I_{t})$ of new tasks is automated, and sufficient
capital for those tasks is instantaneously accumulated.  By contrast,
if the economy experiences a bout of progress that leads to a discrete
mass of tasks suddenly being amenable to automation, the accumulation
of the relevant specific capital may lag behind. The rapid rise of
LLMs at the time of writing may be an example of such a bout.

\begin{figure}
\centering{}\subfloat{\includegraphics[width=1\columnwidth]{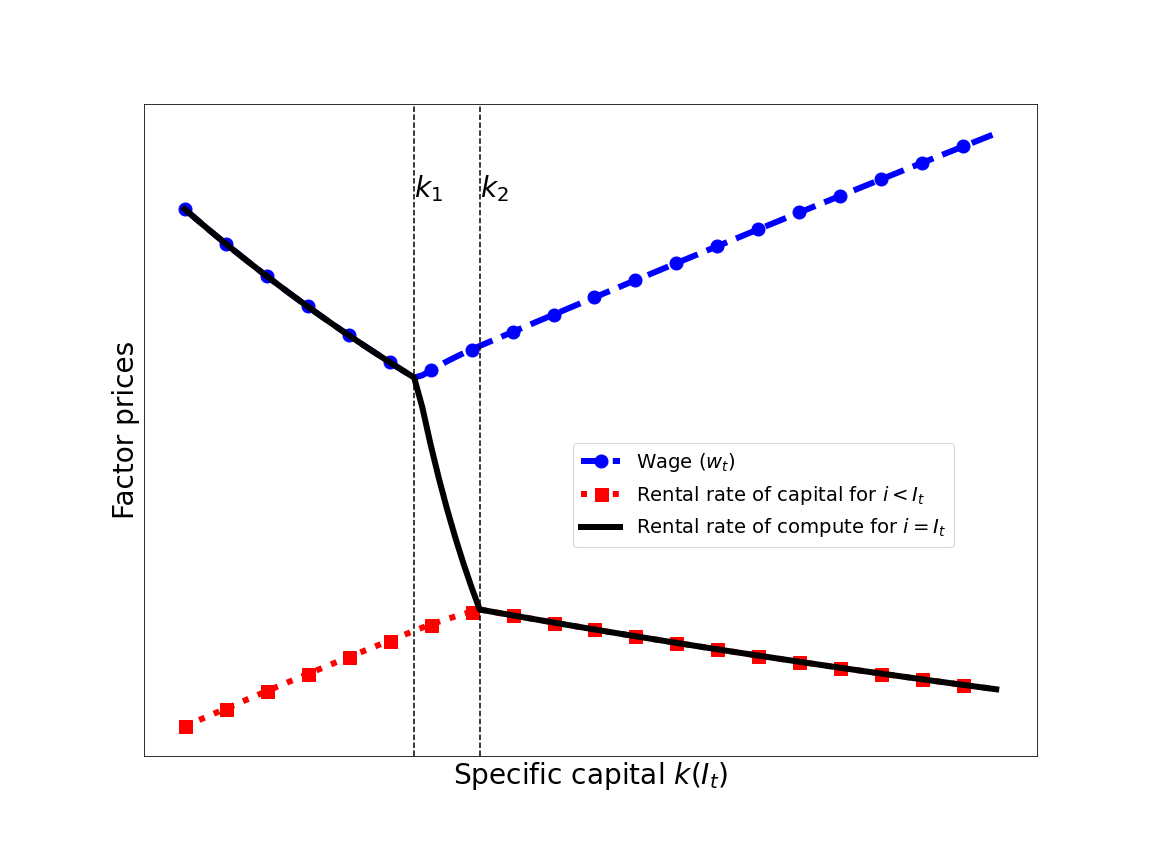}}\caption{\label{fig:specificK}Specific capital and factor prices}
\end{figure}

To illustrate this analytically, assume that a discrete mass of tasks
$\Delta_{t}=\Phi(I_{t})-\Phi(I_{t^{-}})>0$ is automated at time $t$,
and let us interpret the specific capital $k(I_{t})$ required for
these tasks as compute. At time $t$, no compute has been accumulated
yet, $k(I_{t})=0$, so all type-$I_{t}$ tasks are performed by humans
at wage $w$, even though they could technically be automated. Figure
\ref{fig:specificK} illustrates the factor returns as a function
of the accumulation of compute $k(I_{t})$ while holding the inputs
of labor and other capital constant: at first, $k(I_{t})=0$, and
the economy starts out at the left side of the figure where labor
and compute are, at the margin, perfect substitutes so the returns
on the two are equated. The rental rate on traditional capital is
comparatively low.

Over time, compute capital $k(I_{t})$ is accumulated and progressively
substitutes for labor. As long as compute remains below the first
threshold $k(I_{t})<k_{1}$, illustrated by the first vertical line
in the figure, labor and compute remain perfect substitutes at the
margin, but wages $w$ decline with the addition of more compute,
whereas the return on traditional capital rises. Within this region,
all capital investment goes into compute. Once sufficient compute
($k_{1}$) is accumulated so that all humans are replaced from type-$I_{t}$
tasks, all labor is allocated to the remaining unautomated tasks with
$i>I_{t}$, and the marginal product of compute decouples from wages.
All capital investment continues to be devoted to compute; wages $w_{t}$
and the returns to traditional capital rise whereas the return on
compute $k(I_{t})$ declines sharply until it reaches the marginal
product of all other types of capital. This is the middle region between
$k_{1}$ and $k_{2}$ where only the return on compute, captured by
the dotted curve, is decreasing. Once the second threshold is passed,
the marginal product of compute and other capital is equated. Any
additional capital investment is spread proportionately across all
types of specific capital $k(i),i\leq I$, and leads to a decline
in the return on capital. In summary, the race between automation
and capital accumulation leads to a non-monotonic response of wages,
depicted by the blue curve marked with dots, and the returns to traditional
capital, depicted by the red curve marked with squares. 

The following proposition characterizes the thresholds for the amount
of specific capital and summarizes the non-monotonic response of factor
prices to the accumulation of this specific capital analytically.
\begin{prop}[Specific capital and factor returns]
\label{prop:specificK}Suppose that the current amount of the specific
capital is given by $k(I_{t})$. There are threshold values $k_{1}$
and $k_{2}>k_{1}$ such that (i) if $k(I_{t})<k_{1}$ then the wage
decreases and the rental rate of the traditional capital increases
with $k(I_{t})$, (ii) if $k_{1}\leq k(I_{t})<k_{2}$ then both the
wage and the rental rate of traditional capital increase with $k(I_{t})$,
and (iii) if $k(I_{t})\geq k_{2}$ then specific capital $k(I_{t})$
is only accumulated alongside traditional capital, and the wage increases
with capital accumulation.
\end{prop}

\begin{proof}
See appendix.
\end{proof}
In summary, rapid advances in automation may lead to episodes in which
certain types of specific capital (like compute) may exhibit very
high returns, but since capital is reproducible, the resulting accumulation
of specific capital will ultimately dissipate the excess returns.
The implication is that after an adjustment period, specific capital
for newly automated processes will be just another form of capital
earning the market rate of return.

\section{Conclusions}

This paper models the economic impact of the transition torwards artificial
general intelligence on output and wages. We develop a compute-centric
framework that represents work as consisting of tasks that vary in
their computational complexity and study how exponential growth in
computing power will affect automation and the advent of artificial
general intelligence (AGI).

The paper illuminates how different plausible assumptions about the
complexity distribution of tasks across \textquotedbl compute space\textquotedbl{}
translate into dramatically different scenarios for economic outcomes.
If the task distribution has an infinite Pareto tail, reflecting unlimited
complexity of human work, then the we show that wages can rise indefinitely
if the tail is sufficiently thick, as capital accumulation automates
ever more complex tasks but there always remains enough for human
labor. However, if the Pareto tail is too thin, then automation ultimately
outpaces capital accumulation and causes a collapse in wages.

Moreover, if the complexity of tasks humans can perform is bounded,
mirroring computational limits on human cognition, then we demonstrate
that wages would at first surge as machines displace more and more
human labor, but would eventually collapse, even before full AGI is
reached. 

Beyond these scenarios, the paper provides several powerful general
insights. Using the economy's factor price frontier, we show that
the effects of automation follow an inverse U-shape, first increasing
wages by utilizing abundant capital but eventually decreasing wages
due to labor displacement. We show that sufficient capital accumulation
is essential to prevent automation from depressing wages. Adding fixed
factors like land causes wages to eventually decline. Yet automating
innovation itself can restart wage growth after an initial automation-driven
collapse.

The novel compute-centric approach opens up a new perspective for
analyzing the economic impact of artificial intelligence. Interesting
next steps include incorporating labor and capital adjustment costs,
modeling endogenous innovation, analyzing distributional impacts more
fully, studying macroeconomic dynamics and policies, and evaluating
the possibility of an intelligence explosion with AGI.

By presenting several rigorous scenarios for how the transition to
AGI may unfold, we hope that this paper will make an important contribution
to enabling economists, policymakers and the public to examine alternative
futures and to prepare for the technological transformations on the
horizon.

\bibliographystyle{apalike}
\bibliography{AGI}

\appendix
\newpage{}

\section{Proofs and Additional Results}

\subsection{Proof of Proposition \ref{prop:race-between-I-and-K}\label{subsec:Proofs-of-Proposition-race-I-and-K}}

To begin with, we show that the growth of capital stock is approximately
exponential at some constant rate. If the economy is in region 1,
the production function is CES. After some algebra, the growth rate
of capital is 
\begin{align*}
\frac{\dot{K}_{t}}{K_{t}} & =\frac{s_{t}A\bigg(K_{t}^{\frac{\sigma-1}{\sigma}}\Phi_{t}^{\frac{1}{\sigma}}+L^{\frac{\sigma-1}{\sigma}}(1-\Phi_{t})^{\frac{1}{\sigma}}\bigg)^{\frac{\sigma}{\sigma-1}}}{K_{t}}-\delta\\
 & =s_{t}A\bigg(\Phi_{t}^{\frac{1}{\sigma}}+K_{t}{}^{-\frac{\sigma-1}{\sigma}}L{}^{\frac{\sigma-1}{\sigma}}(1-\Phi_{t})^{\frac{1}{\sigma}}\bigg)^{\frac{\sigma}{\sigma-1}}-\delta\\
 & =s_{t}A\bigg(1+L^{\frac{\sigma-1}{\sigma}}\bigg/K_{t}{}^{\frac{\sigma-1}{\sigma}}\bigg(\frac{\Phi_{t}}{1-\Phi_{t}}\bigg)^{\frac{1}{\sigma}}\bigg)^{\frac{\sigma}{\sigma-1}}\cdot\Phi_{t}^{\frac{1}{\sigma-1}}-\delta
\end{align*}
Note that the growth rate of capital stock depends on the behavior
of $\Omega_{t}\equiv K_{t}{}^{\frac{\sigma-1}{\sigma}}\bigg(\frac{\Phi_{t}}{1-\Phi_{t}}\bigg)^{\frac{1}{\sigma}}$.
In particular, $\frac{\Phi_{t}}{1-\Phi_{t}}$ grows approximately
at an exponential rate under the Pareto assumption because
\begin{align*}
\frac{\Phi_{t}}{1-\Phi_{t}} & =\frac{1-I_{t}^{-\lambda}}{I_{t}^{-\lambda}}\\
 & I_{t}^{\lambda}-1\\
 & =I_{0}^{\lambda}e^{\lambda gt}-1\\
 & \approx I_{0}^{\lambda}e^{\lambda gt}
\end{align*}
We consider three cases and see whether each case is consistent with
derivation above. First, suppose $\Omega_{t}\rightarrow\infty$. Then
capital stock must grow at a sufficiently low rate at least for large
$t$. That is, the long-run growth rate of capital $g_{K}$ must satisfy
\[
g_{K}<\frac{\lambda g}{1-\sigma}
\]
In this case, the growth rate of capital is 
\begin{align*}
\frac{\dot{K}_{t}}{K_{t}} & \rightarrow sA\bigg(1+0\bigg)^{\frac{\sigma}{\sigma-1}}\cdot1-\delta\\
 & =s^{\infty}A-\delta
\end{align*}
where $s^{\infty}$ is the long-run savings rate defined in Lemma
\ref{lem:long-run-Ramsey}. Then we have $g_{K}=s^{\infty}A-\delta<\frac{\lambda g}{1-\sigma}$
and thus the following upper bound on the long-run savings rate
\[
s^{\infty}<\frac{1}{A}\bigg(\frac{\lambda g}{1-\sigma}+\delta\bigg)
\]

Secondly, consider the case where $\Omega_{t}\rightarrow0$. Then
capital stock must grow at a rate such that 
\[
g_{K}>\frac{\lambda g}{1-\sigma}
\]
In this case, the growth rate of capital converges to a negative value
\begin{align*}
\frac{\dot{K}_{t}}{K_{t}} & \rightarrow sA\cdot0\cdot1-\delta\\
 & =-\delta
\end{align*}
But this contradicts the lower bound on $g_{K}$ and puts an upper
bound on the growth rate of capital stock.

Lastly, suppose $\Omega_{t}$ converges to a nonzero constant. Then
it must be the case that
\[
g_{K}=\frac{\lambda g}{1-\sigma}
\]
which requires $\Omega_{t}\rightarrow\Omega$ where $\Omega$ is some
constant satisfying $\frac{\dot{K}_{t}}{K_{t}}\rightarrow s^{\infty}A\bigg(1+L^{\frac{\sigma-1}{\sigma}}\bigg/\Omega\bigg)^{\frac{\sigma}{\sigma-1}}-\delta=\frac{\lambda g}{1-\sigma}$.
In the three cases, capital stock grows asymptotically at either $s^{\infty}A-\delta$
or $\lambda g/(1-\sigma)$.

To characterize the long-run labor income share, note that the labor
share is 
\begin{align*}
LS_{t} & =\frac{w_{t}L}{Y_{t}}\\
 & =\frac{A^{\frac{\sigma-1}{\sigma}}(Y_{t}/L)^{\frac{1}{\sigma}}(1-\Phi_{t})^{\frac{1}{\sigma}}L}{Y_{t}}\\
 & =A^{\frac{\sigma-1}{\sigma}}Y_{t}^{\frac{1}{\sigma}-1}L^{1-\frac{1}{\sigma}}(1-\Phi_{t})^{\frac{1}{\sigma}}\\
 & =A^{\frac{\sigma-1}{\sigma}}Y_{t}^{-\frac{\sigma-1}{\sigma}}L^{\frac{\sigma-1}{\sigma}}(1-\Phi_{t})^{\frac{1}{\sigma}}\\
 & =A^{\frac{\sigma-1}{\sigma}}A^{-\frac{\sigma-1}{\sigma}}\bigg(K_{t}^{\frac{\sigma-1}{\sigma}}\Phi_{t}^{\frac{1}{\sigma}}+L^{\frac{\sigma-1}{\sigma}}(1-\Phi_{t})^{\frac{1}{\sigma}}\bigg)^{-1}L^{\frac{\sigma-1}{\sigma}}(1-\Phi_{t})^{\frac{1}{\sigma}}\\
 & =\bigg(K_{t}^{\frac{\sigma-1}{\sigma}}\Phi_{t}^{\frac{1}{\sigma}}+L^{\frac{\sigma-1}{\sigma}}(1-\Phi_{t})^{\frac{1}{\sigma}}\bigg)^{-1}L^{\frac{\sigma-1}{\sigma}}(1-\Phi_{t})^{\frac{1}{\sigma}}\\
 & =\frac{L^{\frac{\sigma-1}{\sigma}}(1-\Phi_{t})^{\frac{1}{\sigma}}}{K_{t}^{\frac{\sigma-1}{\sigma}}\Phi_{t}^{\frac{1}{\sigma}}+L^{\frac{\sigma-1}{\sigma}}(1-\Phi_{t})^{\frac{1}{\sigma}}}\\
 & =\frac{L^{\frac{\sigma-1}{\sigma}}}{K_{t}^{\frac{\sigma-1}{\sigma}}\bigg(\frac{\Phi_{t}}{1-\Phi_{t}}\bigg){}^{\frac{1}{\sigma}}+L^{\frac{\sigma-1}{\sigma}}}\\
 & =\frac{L^{\frac{\sigma-1}{\sigma}}}{\Omega_{t}+L^{\frac{\sigma-1}{\sigma}}}
\end{align*}
In the first case, $\Omega_{t}\rightarrow1$ and so $LS_{t}\rightarrow0$.
In the second case, $\Omega_{t}\rightarrow0$ and so $LS_{t}\rightarrow1$.
Lastly, in the third case, $\Omega_{t}\rightarrow\bar{K}^{\frac{\sigma-1}{\sigma}}\bar{I}^{\frac{1}{\sigma}}$.
Since \textbf{$s^{\infty}A\bigg(1+L^{\frac{\sigma-1}{\sigma}}\bigg/\Omega\bigg)^{\frac{\sigma}{\sigma-1}}-\delta=\frac{\lambda g}{1-\sigma}$},
it follows that
\begin{align*}
s^{\infty}A\bigg(1+L^{\frac{\sigma-1}{\sigma}}\bigg/\Omega\bigg)^{\frac{\sigma}{\sigma-1}}-\delta & =\frac{\lambda g}{1-\sigma}\\
\bigg(1+L^{\frac{\sigma-1}{\sigma}}\bigg/\Omega\bigg)^{\frac{\sigma}{\sigma-1}} & =\frac{1}{s^{\infty}A}\bigg(\frac{\lambda g}{1-\sigma}+\delta\bigg)\\
\bigg(\frac{\Omega+L^{\frac{\sigma-1}{\sigma}}}{\Omega}\bigg)^{\frac{\sigma}{\sigma-1}} & =\frac{1}{s^{\infty}A}\bigg(\frac{\lambda g}{1-\sigma}+\delta\bigg)\\
\bigg(\frac{\Omega}{\Omega+L^{\frac{\sigma-1}{\sigma}}}\bigg)^{\frac{\sigma}{\sigma-1}} & =\frac{s^{\infty}A}{\frac{\lambda g}{1-\sigma}+\delta}\\
\bigg(1-\frac{L^{\frac{\sigma-1}{\sigma}}}{\Omega+L^{\frac{\sigma-1}{\sigma}}}\bigg)^{\frac{\sigma}{\sigma-1}} & =\frac{s^{\infty}A}{\frac{\lambda g}{1-\sigma}+\delta}\\
1-\frac{L^{\frac{\sigma-1}{\sigma}}}{\Omega+L^{\frac{\sigma-1}{\sigma}}} & =\bigg[\frac{s^{\infty}A}{\frac{\lambda g}{1-\sigma}+\delta}\bigg]^{\frac{\sigma-1}{\sigma}}\\
\frac{L^{\frac{\sigma-1}{\sigma}}}{\Omega+L^{\frac{\sigma-1}{\sigma}}} & =1-\bigg[\frac{s^{\infty}A}{\frac{\lambda g}{1-\sigma}+\delta}\bigg]^{\frac{\sigma-1}{\sigma}}
\end{align*}
Therefore, $LS_{t}\rightarrow1-\bigg[\frac{s^{\infty}A}{\frac{\lambda g}{1-\sigma}+\delta}\bigg]^{\frac{\sigma-1}{\sigma}}=1-\bigg[\frac{(A-\rho-\delta+\eta\delta)/\eta}{\frac{\lambda g}{1-\sigma}+\delta}\bigg]^{\frac{\sigma-1}{\sigma}}$.

If the economy starts in region 1 then the economy stays in region
1 as long as
\[
\frac{\dot{\hat{I}}_{t}}{\hat{I}_{t}}\geq\frac{\dot{I}_{t}}{I_{t}}
\]
That is, the threshold grows faster than the automation index. Under
the Pareto assumption, the inequality is equivalent to
\[
\frac{\dot{K}_{t}}{K_{t}}\geq\frac{1+K_{t}/L}{K_{t}/L}\cdot\lambda g
\]
where $\frac{1+K_{t}/L}{K_{t}/L}$ converges to one from above. Thus,
the above inequality delivers a lower bound on the savings rate for
the economy to asymptotically stay in region 1:
\begin{equation}
s^{\infty}A-\delta>\lambda g\label{eq:lower-bound-on-s}
\end{equation}
The inequality ensures that capital accumulation is sufficiently fast
compared to automation. If it is violated then $I_{t}$ crosses $\hat{I}_{t}$
eventually and wages collapse to $A$.

To further examine how capital accumulation and automation shape the
asymptotic behavior of wages, consider the growth rate of wages
\[
\frac{\dot{w}_{t}}{w_{t}}=\frac{1}{\sigma}\frac{\dot{Y}_{t}}{Y_{t}}-\frac{1}{\sigma}\frac{\dot{\Phi}_{t}}{1-\Phi_{t}}
\]
derived from $w_{t}=F_{L}$ and
\begin{align*}
\log F_{L} & =\log A+\frac{1}{\sigma}\log(Y/A)-\frac{1}{\sigma}\log L+\frac{1}{\sigma}\log(1-\Phi(I))\\
\frac{d\log F_{L}}{dt} & =\frac{1}{\sigma}\frac{d\log Y}{dt}+\frac{1}{\sigma}\frac{d\log(1-\Phi(I))}{dt}\\
 & =\frac{1}{\sigma}\frac{d\log Y}{dt}+\frac{1}{\sigma}\frac{1}{1-\Phi}(-\dot{\Phi})
\end{align*}
. The above equation shows that the growth rate consists of output
growth and the displacement effect of automation. Note that the growth
rate of output is
\[
\frac{\dot{Y}_{t}}{Y_{t}}=S_{K}\frac{\dot{K}_{t}}{K_{t}}+\frac{1}{1-\sigma}\frac{\dot{\Phi}_{t}}{1-\Phi_{t}}\bigg(S_{L}-S_{K}\frac{1-\Phi_{t}}{\Phi_{t}}\bigg)
\]
where $S_{K}\equiv\frac{K_{t}^{\frac{\sigma-1}{\sigma}}\Phi_{t}^{\frac{1}{\sigma}}}{K_{t}^{\frac{\sigma-1}{\sigma}}\Phi_{t}^{\frac{1}{\sigma}}+L^{\frac{\sigma-1}{\sigma}}(1-\Phi_{t})^{\frac{1}{\sigma}}}$
and $S_{L}\equiv\frac{L^{\frac{\sigma-1}{\sigma}}(1-\Phi_{t})^{\frac{1}{\sigma}}}{K_{t}^{\frac{\sigma-1}{\sigma}}\Phi_{t}^{\frac{1}{\sigma}}+L^{\frac{\sigma-1}{\sigma}}(1-\Phi_{t})^{\frac{1}{\sigma}}}$,
omitting time subscripts for notational simplicity. The first term
is growth due to capital accumulation and the second term is growth
due to the productivity effect of automation. The wage growth rate
is then
\begin{equation}
\frac{\dot{w}_{t}}{w_{t}}=\frac{1}{\sigma}\bigg[S_{K}\frac{\dot{K}_{t}}{K_{t}}+\frac{1}{1-\sigma}\frac{\dot{\Phi}_{t}}{1-\Phi_{t}}\bigg(S_{L}-S_{K}\frac{1-\Phi_{t}}{\Phi_{t}}\bigg)-\frac{\dot{\Phi}_{t}}{1-\Phi_{t}}\bigg]\label{eq:wage-growth-1}
\end{equation}
That is, wages rise as long as capital accumulation and the productivity
effect dominate the displacement effect. In fact, this is another
version of (\ref{eq:balanced-savings-rate}), which can be seen by
setting $\dot{w}_{t}=0$, and tells us what determines the growth
rate of wages. Under the Pareto assumption, we have
\[
\frac{\dot{w}_{t}}{w_{t}}=\frac{1}{\sigma}\bigg[S_{K}\frac{\dot{K}_{t}}{K_{t}}+\frac{1}{1-\sigma}\cdot\lambda g\cdot\bigg(S_{L}-S_{K}\frac{1-\Phi_{t}}{\Phi_{t}}\bigg)-\lambda g\bigg]
\]
where $\lambda g$ is the rate of automation adjusted by the decay
rate of the fraction of tasks for labor. 

Notice that 
\[
S_{K}=\frac{(*)}{(*)+L^{\frac{\sigma-1}{\sigma}}}
\]
If $\lambda g>(1-\sigma)(s^{\infty}A-\delta)$ (i.e. the first case
in the beginning of the proof) then $S_{K}\rightarrow1$ since $(*)\rightarrow\infty$.
Note that
\begin{align*}
\frac{\dot{w}_{t}}{w_{t}} & =\frac{1}{\sigma}\bigg[S_{K}\frac{\dot{K}_{t}}{K_{t}}+\frac{1}{1-\sigma}\cdot\lambda g\cdot\bigg(S_{L}-S_{K}\frac{1-\Phi_{t}}{\Phi_{t}}\bigg)-\lambda g\bigg]
\end{align*}
 As $t\rightarrow\infty$, the growth rate of wages converges as
follows
\begin{align*}
\frac{\dot{w}_{t}}{w_{t}} & \rightarrow\frac{1}{\sigma}\bigg[1\cdot(s^{\infty}A-\delta)+\frac{1}{1-\sigma}\cdot\lambda g\cdot\bigg(0-1\cdot0\bigg)-\lambda g\bigg]\\
 & =\frac{1}{\sigma}\bigg[s^{\infty}A-\delta-\lambda g\bigg]
\end{align*}
Thus, if $s^{\infty}A-\delta>\lambda g$ (but $s^{\infty}A-\delta<\lambda g/(1-\sigma)$)
then wages grow exponentially at an asymptotic rate $\frac{1}{\sigma}\bigg[s^{\infty}A-\delta-\lambda g\bigg]=\frac{1}{\sigma}\bigg[\frac{A-\rho-\delta}{\eta}-\lambda g\bigg]$. 

In the case where capital stock asymptotically grows at $\lambda g/(1-\sigma)$,
$S_{K}$ converges to one. As a result, the growth rate of wages converges
as follows
\begin{align*}
\frac{\dot{w}_{t}}{w_{t}} & \rightarrow\frac{1}{\sigma}\bigg[S_{K}\frac{\lambda g}{1-\sigma}+\frac{1}{1-\sigma}\cdot\lambda g\cdot\bigg(1-S_{K}-S_{K}\cdot0\bigg)-\lambda g\bigg]\\
 & =\frac{1}{\sigma}\bigg[S_{K}\frac{\lambda g}{1-\sigma}+\frac{\lambda g}{1-\sigma}\cdot(1-S_{K})-\lambda g\bigg]\\
 & =\frac{1}{\sigma}\bigg[\frac{\lambda g}{1-\sigma}-\lambda g\bigg]\\
 & =\frac{\lambda g}{1-\sigma}
\end{align*}
\begin{align*}
\therefore\frac{\dot{w}_{t}}{w_{t}} & \rightarrow\frac{\lambda g}{1-\sigma}
\end{align*}
If $s^{\infty}A-\delta\leq\lambda g$ then wages decline until the
automation index crosses the threshold and collapse to $A$, as (\ref{eq:lower-bound-on-s})
indicates.

\subsection{Proof of Proposition \ref{prop:specificK}}
\begin{proof}[Proof of Proposition \ref{prop:specificK}]
The production function can be written as
\begin{align*}
Y_{t} & =A\bigg(K_{t}^{\frac{\sigma-1}{\sigma}}\Phi(I_{t-})^{\frac{1}{\sigma}}+K(I)_{t}^{\frac{\sigma-1}{\sigma}}\Delta_{t}^{\frac{1}{\sigma}}+L^{\frac{\sigma-1}{\sigma}}(1-\Phi(I_{t}))^{\frac{1}{\sigma}}\bigg)^{\frac{\sigma}{\sigma-1}}\\
F_{K} & =A\bigg(K_{t}^{\frac{\sigma-1}{\sigma}}\Phi(I_{t-})^{\frac{1}{\sigma}}+K(I)_{t}^{\frac{\sigma-1}{\sigma}}\Delta_{t}^{\frac{1}{\sigma}}+L^{\frac{\sigma-1}{\sigma}}(1-\Phi(I_{t}))^{\frac{1}{\sigma}}\bigg)^{\frac{1}{\sigma-1}}K_{t}^{-\frac{1}{\sigma}}\Phi(I_{t-})^{\frac{1}{\sigma}}\\
F_{K(I_{t})} & =A\bigg(K_{t}^{\frac{\sigma-1}{\sigma}}\Phi(I_{t-})^{\frac{1}{\sigma}}+K(I)_{t}^{\frac{\sigma-1}{\sigma}}\Delta_{t}^{\frac{1}{\sigma}}+L^{\frac{\sigma-1}{\sigma}}(1-\Phi(I_{t}))^{\frac{1}{\sigma}}\bigg)^{\frac{1}{\sigma-1}}K(I)_{t}^{-\frac{1}{\sigma}}\Delta_{t}^{\frac{1}{\sigma}}\\
F_{L} & =A\bigg(K_{t}^{\frac{\sigma-1}{\sigma}}\Phi(I_{t-})^{\frac{1}{\sigma}}+K(I)_{t}^{\frac{\sigma-1}{\sigma}}\Delta_{t}^{\frac{1}{\sigma}}+L^{\frac{\sigma-1}{\sigma}}(1-\Phi(I_{t}))^{\frac{1}{\sigma}}\bigg)^{\frac{1}{\sigma-1}}L^{-\frac{1}{\sigma}}(1-\Phi(I_{t}))^{\frac{1}{\sigma}}
\end{align*}
If $F_{L}<F_{K(I_{t})}$, then the specific capital and labor are
perfectly substitutable. That is,
\begin{align*}
\frac{K(I_{t})}{L} & <\frac{\Delta_{t}}{1-\Phi(I_{t})}\\
k(I_{t})=\frac{K(I_{t})}{\Delta_{t}} & <\frac{L}{1-\Phi(I_{t})}
\end{align*}
in which case, the production function can be written as
\[
Y_{t}=A\bigg(K_{t}^{\frac{\sigma-1}{\sigma}}\Phi(I_{t-})^{\frac{1}{\sigma}}+(K(I_{t})+L)^{\frac{\sigma-1}{\sigma}}(1-\Phi(I_{t-}))^{\frac{1}{\sigma}}\bigg)^{\frac{\sigma}{\sigma-1}}
\]
If $F_{K}>F_{K(I_{t})}$, then the specific capital and the traditional
capital are perfectly substitutable. That is,
\begin{align*}
\frac{K(I_{t})}{K_{t}} & >\frac{\Delta_{t}}{\Phi(I_{t-})}\\
k(I_{t})=\frac{K(I_{t})}{\Delta_{t}} & >\frac{K_{t}}{\Phi(I_{t-})}
\end{align*}
in which case, the production function can be written as
\[
Y_{t}=A\bigg((K_{t}+K(I_{t}))^{\frac{\sigma-1}{\sigma}}\Phi(I_{t})^{\frac{1}{\sigma}}+L{}^{\frac{\sigma-1}{\sigma}}(1-\Phi(I_{t}))^{\frac{1}{\sigma}}\bigg)^{\frac{\sigma}{\sigma-1}}
\]
\end{proof}

\end{document}